\documentclass[preprint]{elsarticle}
\usepackage[cmex10]{amsmath}
\usepackage{prooftree}
\usepackage{stmaryrd}
\usepackage{color}
\usepackage{cmll}
\usepackage{amsmath,amssymb}
\usepackage{esint}
\usepackage{psfrag}

\def\Bot{{{\bot}\mskip-11mu{\bot}}}

\newcommand{\N} {\mathbb{N}}

\newcommand{\Val}[1]{\llbracket #1 \rrbracket}

\newcommand{\setp}[2]{\{\,\, #1 \,\,|\,\, #2 \,\,\}}

\newcommand{\M}{\mathcal{M}}

\newcommand{\interp}[1]{| #1 |}
\newcommand{\norm}[1]{\| #1 \|}
\def\unit{{\bf 1}}

\newcommand{\Mp}[1]{{\bf M[}#1{\bf ]}}

\def\Terms{\mathcal{T}}

\def\redu{\longrightarrow}

\def\zero{{\bf 0}}

\newcommand{\class}[1]{\textbf{\textsc{#1}}}
\def\dai{\maltese}

\def\dbot{{\bot\bot}}

\newcommand{\FTrans}[1]{{#1^\bullet}}
\newcommand{\forcNeg}[1]{{{#1}^\ominus}}

\newcommand{\com}[2]{\langle \, #1 \, | \, #2 \, \rangle}
\newcommand{\coml}[1]{\langle #1 |}
\newcommand{\comr}[1]{ | #1 \rangle}
\def\redu{\rightarrow}
\def\redh{\rightarrow_0}
\def\nredh{\nrightarrow_0}
\def\redi{\rightarrow_i}

\def\redMu{\rightarrow_\mu}

\def\redBeta{{\rightarrow_\beta}}

\def\pBeta{{p_\beta}}

\def\Time{{\text{{\bf Time}}}}

\newcommand{\BotTime}{{\Bot_{{\Time}}}}

\def\PosTerms{{\mathcal{T}_+^0}}
\def\NegTerms{{\mathcal{T}_-^0}}
\def\ComSet{{\mathcal{C}}}
\def\dbot{{\bot\bot}}
\newcommand{\BV}[1]{{#1}_{\mathbb{V}}}

\def\daiPos{{\dai_+}}
\def\daiNeg{{\dai_-}}

\def\FVdash{{\,\Vdash_f\,}}

\def\MAL{{\class{MAL}} }
\def\HOMAL{{\bf $\text{MAL}\omega$} }

\def\HOF{{\bf $\text{F}\omega$ }}
\def\HOPA{{\bf $\text{PA}\omega$ }}

\def\HOSAL{{\bf $\text{SAL}\omega$} }

\def\LFOC{{$\text{{\bf L}}_{foc}$}}

\def\vdashSAL{\vdash_{\text{SAL}\omega}}

\def\defequ{{\equiv}}

\def\BotTime{{\Bot_{Time}}}
\def\Can{{\mathcal{D}_{can}}}
\def\CanNeg{{\mathcal{D}^\ominus_{can}}}
\def\CanPos{{\mathcal{D}^\oplus_{can}}}

\newcommand{\subtitle}[1]{	{\noindent\textbf{#1  \quad\textemdash\quad}}}

\def\sp{\,\,}
\def\delb{\sp|\sp}

\def\Dpos{{\mathcal{D^{\oplus}}}}
\def\Dneg{{\mathcal{D^{\ominus}}}}
\def\congE{\cong}

\def\sortint{\iota}
\def\sortpropPos{{o^{+}}}
\def\sortpropNeg{{o^{-}}}
\def\sortprop{o}
\def\sortpropW{o^{*}}

\newcommand{\Inte}[1]{{\norm{#1}}}
\def\foZero{{\bf 0}}
\def\foSucc{{\bf s}}
\newcommand{\foN}[1]{\mathbf{#1}}

\def\Kpos{{\mathcal{K}_+}}
\def\Kneg{{\mathcal{K}_-}}

\newdimen\boxfigwidth 

\def\bigbox{\begingroup
  \boxfigwidth=\hsize
  \advance\boxfigwidth by -2\fboxrule
  \advance\boxfigwidth by -2\fboxsep
  \setbox4=\vbox\bgroup\hsize\boxfigwidth
  \hrule height0pt width\boxfigwidth\smallskip%
  \linewidth=\boxfigwidth
}
\def\endbigbox{\smallskip\egroup\fbox{\box4}\endgroup}

\newcounter{counter}

\newdefinition{defn}[counter]{Definition}
\newdefinition{rmk}[counter]{Remark}
\newdefinition{rmks}[counter]{Remarks}
\newdefinition{exa}[counter]{Example}
\newtheorem{lem}[counter]{Lemma}
\newtheorem{prop}[counter]{Property}
\newtheorem{thm}[counter]{Theorem}
\newproof{proof}{Proof}

\begin{document}

\title{Quantitative classical realizability}

\author[rvt]{Alo\"is Brunel}
\ead{alois.brunel@ens-lyon.org}

\address[rvt]{LIPN - UMR CNRS 7030 - Universit\'e Paris 13, Villetaneuse, France}

\begin{abstract}
Introduced by Dal Lago and Hofmann, quantitative realizability is a 
technique used to define models for logics based on Multiplicative Linear Logic. 
A particularity is that functions are interpreted as bounded time computable functions.
It has been used to give new and uniform proofs of soundness of several type systems
with respect to certain time complexity classes. 
We propose a reformulation of their ideas in the setting of Krivine's classical realizability. 
The framework obtained generalizes Dal Lago and Hofmann's realizability, and reveals
deep connections between quantitative realizability and a linear variant of Cohen's forcing. 
\end{abstract}

\maketitle

\section{Introduction}

Ever since its introduction by J.L Krivine \cite{krivine-realizability}, the
theory of classical realizability has raised a growing interest. Initially
designed to study the computational content of classical proofs through the
Curry-Howard correspondence, it has led to promising results in various fields.
One could mention the recent advances \cite{krivine2010realizability2} made
by Krivine in the elaboration of new models of the \class{ZF} axiomatic set theory.
Another success has been its use to define and justify a classical extraction procedure for the proof assistant Coq \cite{miquel2007classical}.\newline

\subtitle{Forcing} Forcing is a technique designed by Cohen \cite{cohen1963independence}
to prove the independence of the Continuum Hypothesis (CH) from \class{ZFC}. The idea
is to define a formula transformation which turns every formula $A$ into a 
new one noted $p\Vdash A$, where $p$ is a \textit{forcing condition}. By choosing
a suitable set of forcing conditions, one can prove the statement $p\Vdash \neg CH$. 
It has been recently shown by Krivine \cite{krivine2010realizability} that combining
classical realizability and forcing is possible. This construction can be seen
as a generalization of forcing iteration and makes possible a study of forcing
through the Curry-Howard isomorphism: Krivine
has shown that the forcing technique 
not only provides a logical translation but also a program transformation. 
Following that work, Miquel \cite{Miquel2011Forcing} has introduced an abstract
machine (the Krivine Forcing Abstract Machine, or \class{KFAM}) that internalizes
the computational behavior of programs obtained via this transformation. 
One remarkable feature of this machine is that it provides sophisticated
programming features like memory cells or program execution tracing. \newline

\subtitle{Resource sensitive realizability} Realizability techniques
have also been fruitfully applied to \textit{implicit complexity}. This research field aims at providing
machine-independent characterizations of complexity classes (such as polynomial time or
logspace functions). One of the possible approaches is to use linear logic based
type systems to constrain programs enough so that they enjoy bounded-time
normalization properties. Proving these properties can be achieved using semantic techniques.
Following different works \cite{DBLP:journals/apal/Hofmann00,Hofmann2004121},
Dal Lago and Hofmann have introduced in \cite{DalLago20112029} a \textit{quantitative} (another word
for \textit{resource sensitive}) framework based on Kleene realizability \cite{kleene1969formalized}.
One of the crucial ideas behind Dal Lago and Hofmann's work is to consider bounded-time 
$\lambda$-terms as realizers. Bounds are described using elements of a \textit{resource monoid}. 
No matter what resource monoid is chosen, their framework always yields a model of second-order Multiplicative 
Affine Logic (\class{MAL}). Various systems extending \class{MAL} are then dealt with by choosing a suitable 
resource monoid, while the basic realizability constructions are unchanged. This work
has offered new and uniform proofs of the soundness theorems for
\class{LAL}, \class{EAL},  \class{SAL} and \class{BLL} with respect to the associated complexity 
classes \cite{dalagobll,dal2010semantic,DalLago20112029}. In \cite{brunel23church}, Terui and the author 
gave a new characterization of the complexity class \class{FP} (the functions computable in
polynomial time) and used a variant of Dal Lago and Hofmann's realizability to show 
the soundness part of this result.\newline

The present work aims at applying methodology and tools coming from classical
realizability to generalize the framework proposed by Dal Lago and Hofmann, and
to reveal deep connections between quantitative realizability and
forcing techniques. \newline

\subtitle{Quantitative classical realizability} We propose a new quantitative framework, 
based on Munch's classical realizability
for focalising system \class{L} (or \LFOC) \cite{MunchReal}, 
a term calculus for classical logic \class{LC} \cite{CambridgeJournals:4439772}. We extend
this realizability using the notion of \textit{quantitative monoid}, which
derives from the \textit{resource monoid} structure introduced by Dal Lago and Hofmann.
We show that, whatever the quantitative monoid, this framework always
gives rise to a model of the Multiplicative Affine fragment of 
Higher-order Classical Arithmetic (abbreviated \class{MAL}$\omega$). By choosing
different quantitative monoids, we obtain models of logics extending \HOMAL.
Because all resource monoids in the sense of \cite{DalLago20112029} are 
also quantitative monoids, we can in principle obtain models for 
all the systems treated in \cite{dal2010semantic,DalLago20112029},
although we only exhibit a model of Soft Affine Logic (\class{SAL}) \cite{baillot2004soft}.\newline

\subtitle{Quantitative reducibility candidates} By carefully setting
parameters of classical realizability, one can retrieve the notion
of \textit{reducibility candidates} (presented using orthogonality,
as in \cite{girard1987linear,  LengrandMiquelClassical2008,okada1999phase,
pagani2010strong}), which is used to prove normalization properties. 
Similarly, in our setting, we are able to define a quantitative 
extension of this technique, which we call \textit{quantitative reducibility candidates}. 
It allows us to semantically prove complexity properties of programs that are typable in the logic
we interpret. Moreover, because
we work with a term calculus which generalizes both call-by-name and call-by-value
classical $\lambda$-calculi, these complexity properties are transferred for free 
to these calculi. Hence, we are able to retrieve and generalize 
 the bounded-time termination
results proved in \cite{dal2010semantic,DalLago20112029}.\newline

\subtitle{A forcing decomposition} Quantitative classical realizability is deeply connected
with a certain notion of forcing, which we propose to study. 
We formalize inside \HOMAL a forcing transformation on Multiplicative Linear Logic (\class{MAL}) formulas, called 
\textit{linear forcing}. Then, following Miquel's methodology \cite{Miquel2011Forcing}, we
propose an abstract machine designed to execute programs obtained by a specific
linear forcing instance. 
A \textbf{connection lemma} is proved, which shows that composing this instance of linear forcing with 
a non-quantitative realizability built upon this machine always yields a quantitative
realizability model. 
Finally, using this result, we show how quantitative reducibility candidates (restricted
to \class{MAL}) arise from the composition of usual
reducibility candidates with forcing. \newline

\subtitle{Outline} Sections \ref{sec:HOMAL} and \ref{sec:QuantKrivine} introduce \HOMAL 
and its quantitative realizability interpretation. The model
of quantitative reducibility candidates is then defined and used to prove a bounded time termination
property of \HOMAL. In section \ref{sec:other_types}, we show by taking \class{SAL} as an example 
that this interpretation and the corresponding complexity result can be extended
to larger type systems. Finally, we introduce in section \ref{sec:forcing} the 
linear forcing interpretation of \MAL and prove the accompanying decomposition results.

\section{The calculus}\label{sec:HOMAL} In this section, we describe the system \class{MAL}$\omega$. It is based on the
\MAL type system for Munch's focalising system \class{L} \cite{MunchReal},
extended with higher-order quantifications and arithmetical operations. 
Logically, it is a fragment of classical higher-order Peano arithmetic
(abbreviated by \class{PA}$\omega$). 
The syntax of \HOMAL is divided in three
distinct layers: the \textit{terms}, the \textit{type constructors} and
the \textit{kinds}. The language of terms, which we shall use to express both proof-terms and realizers, is based on the multiplicative fragment of \LFOC,
extended with extra instructions. The type constructors layer is an adaptation
of the higher-order terms syntax of \HOPA \cite{Miquel2011Forcing} to linear logic: it can be seen 
as a combination of the languages of \HOPA and classical \HOF \cite{LengrandMiquelClassical2008}. Finally,
kinds are used as a simple type system for type constructors. 

\subsection{Term syntax}\label{subsec:termsyntax}
In what follows, 
\textit{positive variables} and \textit{negative variables} are respectively written $x,y,z, \dots$
and $\alpha,\beta,\gamma,\dots$. We use the symbols $\kappa,\kappa',\dots$ to denote both
positive and negative variables. In the term syntax of \LFOC, in addition to variables, 
six syntactic categories are defined:
\emph{values}, \emph{positive values}, 
\emph{positive terms}, \emph{negative terms}, \emph{terms} and \emph{commands}:

\begin{tabular}{rcclr}
 & & & & \\
 \textbf{variables} & $\kappa,\kappa'$ & $::=$ & $\alpha \delb x$ & \\
 \textbf{values} 	& $V$ & $::=$ & $V_+ \delb t_-$ &  \\
 \textbf{positive values} & $V_+$ & $::=$ & $x\delb (V,V') \delb \{ V\}\delb k_+$ & $(k_+\in \Kpos)$  \\
 \textbf{positive terms} & $t_+$ &  $::=$ & $V_+ \delb \mu \alpha.c$ & \\
 \textbf{negative terms} & $t_-$ & $::=$ & $\alpha\delb \mu (\kappa,\kappa').c \delb \mu\{\kappa\}.c\delb
  	\mu x.c \delb k_-$ &  $(k_-\in \Kneg)$ \\
 \textbf{terms} & $t,u$ & $::=$ & $t_- \delb t_+$\\
 \textbf{commands} & $c$ & $::=$   & $\com{t_+}{t_-}$& \\
 & & & \\
\end{tabular}
\noindent where $\mu (\kappa, \kappa').c$ is not defined if $\kappa = \kappa'$. Moreover terms are always considered modulo $\alpha$-equivalence. We also make an identification between the commands
$\com{t}{u}$ and $\com{u}{t}$. Finally, we associate to every term $t$ its \emph{polarity} $\pi(t) \in \{ -, +\}$ as follows:
$$
\pi(t) = 
\left \{ 
\begin{array}{cl}
 + & \text{if }t \text{ is a positive term} \\
 - & \text{if }t \text{ is a negative term}
 \end{array} \right.
 $$

\begin{rmks}\label{rmk:term_syntax}$ $ 
\begin{enumerate}
\item
Notice that in the definition of the pair construct $(V,V')$, $V$ and $V'$ can be
values of arbitrary polarity (that is, positive or negative). We could have
 made the choice of restricting a pair to positive values. This would not be 
 problematic since we still could used $\{.\}$ to change the polarity of
 values from negative to positive before putting them into a pair.  
\item The term $\{ V \}$ can be seen as a one-tuple and is here to give the possibility
of turning a negative term into a positive one. 
\item This untyped calculus has no linear restriction on the use
of variables. However, such restrictions will appear in the type system. 
\item  The identification of $\com{t}{u}$ and $\com{u}{t}$ accounts
for the involutivity of linear negation.
\end{enumerate}
\end{rmks}

Similarly to \cite{Miquel2011Forcing}, the syntax is parametrized by a set of \textit{positive instructions} $\Kpos$ (which are considered as values) 
and a set of \textit{negative instructions} $\Kneg$. This allows us to extend the language at will, in the spirit
of Krivine's $\lambda_c$-calculus \cite{krivine-realizability}. 

\begin{rmk}
If we want to make a comparison with Miquel's work \cite{Miquel2011Forcing}, our
set $\Kneg$ corresponds to the set of \emph{instructions} while $\Kpos$
corresponds to the set of stack constants. 
\end{rmk}

If $x$ is a term or a command, $FV(x)$ denotes the set of the free variables of $x$. 
In the rest of this paper, the sets of \emph{closed terms}, \textit{closed positive terms}, \textit{closed negative terms} and \textit{closed commands} are denoted  respectively by $\Terms^{0}$, $\PosTerms$,  $\NegTerms$ and $\ComSet$. Moreover the set of values is denoted by $\mathbb{V}$. 

\subsection{Reduction}\label{subsec:reduction}

We now present the operational semantics for the syntax we just defined.
The set of commands is equipped with the following one-step reduction relations $\redMu$
and $\redBeta$:

\begin{tabular}{ccrcl}
 & & & & \\
  ($+$) &  \quad & $\com{\mu \alpha.c}{t_-}$ & $\redMu$ & $c[t_-/\alpha]$\\
  ($-$) & \quad & $\com{V^+}{\mu x.c}$ & $\redMu$ & $c[V^+/x]$\\
  ($\uparrow$) & \quad & $\com{\{V\}}{\mu \{\kappa\}.c}$ & $\redBeta$ &  $c[V/\kappa]$\\
  ($\parr$) & \quad & $\com{(V,V')}{\mu(\kappa,\kappa').c}$ & $\redBeta$ & $c[V/\kappa, V'/\kappa']$\\ 
 & & & & \\
 \multicolumn{5}{l}{($\redBeta$ is defined only if the polarities of the $\kappa$'s and the $V$'s match)}\\
 & & & & \\
\end{tabular}

We pose $\redh \sp = \sp \redBeta \sp \cup \sp \redMu$. 
\begin{rmks}
  $ $ 
  
  \begin{enumerate}
  \item
The grammar defining the term syntax does not prevent 
ill-formed commands to appear. Indeed, consider the command
$\com{\mu  \alpha.c}{\{V\}}$.
$\{V \}$ is a positive term whereas $\alpha$  is a
negative variable. Hence, this command won't reduce. 
The possibility of this kind of ill-formed commands and terms 
will be removed by
typing.
  \item 
  Even if the term syntax does not allow directly to form the pair $(t,u)$
  or the one-tuple $\{t\}$  when $t$ and $u$ are not values, 
  it is possible to \emph{define} these constructions as follows:
  \begin{eqnarray*}
  	  (t,u)  & = & \mu\alpha.\com{t}{\mu \kappa.\com{u}{\mu \kappa'.\com{(\kappa,\kappa')}{\alpha}}}\\
	  \{ t\} &=  & \mu\alpha.\com{t}{\mu \kappa.\com{\{\kappa\}}{\alpha}}
	  \end{eqnarray*}
 where the polarities of $\kappa$ and $\kappa'$ respectively match those of 
 $t$ and $u$. In the case of the pair, this definition reflects an arbitrary choice
 in the order of evaluation of $t$ and $u$ (here from left to right). 
 \end{enumerate}
\end{rmks}

\begin{defn}[Evaluation relation]
Similarly to \cite{Miquel2011Forcing}, we consider an \emph{evaluation relation} to be a
binary relation $\redu$ between commands such that $\redh \subseteq \redu$. 
\end{defn}

In the rest paper, $\redu$ always denotes such an evaluation relation. 
\begin{rmk}
The fact that $\redu$ is not fixed will allow us to consider reduction rules
when we extend the term syntax with new instructions, without loosing
the properties and theorems already proved. 
\end{rmk}


\begin{defn}
Suppose $\rightarrow$ is a binary relation between commands. If
$c$ is normalizing for $\rightarrow\cup\redMu$, then
we define $\Time^{\rightarrow}(c)$ as the number of $\rightarrow$ steps needed
by $c$ to normalize using $\rightarrow$ and $\redMu$.  Otherwise, $\Time^{\rightarrow}(c)$ is undefined. 
\end{defn}


\subsection{Kinds and type constructors} 
Here is exposed the language of \HOMAL types.
We define two syntactic categories: \textit{kinds} and \textit{type constructors} (or
simply \textit{constructors}).  

\begin{tabular}{lrcl}
& & & \\
{\bf Kinds} & $\sigma, \tau$ & ::= & $\iota\sp |\sp o^{+}\sp |\sp o^{-} \sp |\sp\sigma\rightarrow \tau$\\
& & & \\
{\bf Constructors} & $A,B, T, U$ & ::= & $x^\tau\delb {x^\tau}^\bot \delb \lambda x^\tau.T \delb TU$\\
                 & & & $\delb \foZero \delb \foSucc \delb rec_\tau \delb rec_\tau^\bot$\\ 
                 & & & $\delb A\otimes B \delb A \parr B\delb \exists x^\tau.A \delb \forall x^\tau.A$\\
                 & & & $\delb \downarrow A \delb \uparrow A$\\
& & & \\
\end{tabular}

Kinds are a simple type system for constructors: $\iota$ is the kind representing
\emph{individuals}, $\sigma\rightarrow \tau$ is the kind of \emph{functions} from
$\sigma$ to $\tau$, $\sortpropPos$ is the kind of \emph{positive formulas}
and $\sortpropNeg$ the kind of \emph{negative formulas}. 
We denote by $\foN{n}$ the constructor $\foSucc^{n} \foZero$.

\begin{defn}[Involutive negation] The operation $(.)^\bot$ (called \emph{negation}) is only defined on atomic constructors (variables and recursor $rec_\tau$). It is extended as an involutive operation on all kinds as follows:
\begin{center}
\begin{tabular}{rclrcl}
  $\sortpropPos^\bot$ & $=$ & $\sortpropNeg$ &
  $\sortpropNeg^\bot$ & $=$ & $\sortpropPos$\\
  $\iota^\bot$ & $=$ & $\iota$ & 
  $(\sigma \rightarrow \tau)^\bot$ & $=$ & $\sigma \rightarrow \tau^{\bot}$
\end{tabular}
\end{center}
But also on all constructors:
\begin{center}
\begin{tabular}{rclrcl} 
 $(x^\tau)^\bot$ & $=$ & ${x^\tau}^\bot$ & 
  $({x^\tau}^\bot)^\bot$ & $=$ & $x^\tau$\\
  $\foZero^{\bot}$ & $=$ & $\foZero$ & 
  $\foSucc^{\bot}$ & $=$ & $\foSucc$\\
  $(\lambda x^\tau.T)^\bot$ & $=$ & $\lambda x^\tau.(T)^\bot$&
  $(TU)^\bot$ & $=$ & $T^\bot U$\\
  $(rec_\tau)^\bot$ & $=$ & $rec_\tau^\bot$ & 
  $(rec_\tau^{\bot})^\bot$ & $=$ & $rec_\tau$ \\
  $(A\otimes B)^\bot$ & $=$ & $A^\bot \parr B^\bot$ & 
  $(A\parr B)^\bot$ & $=$ & $A^\bot \otimes B^\bot$\\
  $(\forall x^\tau.A)^\bot$ & $=$ & $\exists x^\tau.A^\bot$ & 
  $(\exists x^\tau.A)^\bot$ & $=$ & $\forall x^\tau.A^\bot$\\
  $(\uparrow A)^\bot$ & $=$ & $\downarrow A^\bot$ &
  $(\downarrow A)^\bot$ & $=$ & $\uparrow A^\bot$\\
\end{tabular}
\end{center}
The operation $(.)^{\bot}$ is involutive: for any constructor $T$,
we have $T^{\bot\bot} = T$.
\end{defn}

The rules of figure \ref{fig:HOMALkinds} 
define what it means for a constructor $T$ to be of kind $\tau$ (and we note it $T: \tau$).
When we write $T : \sortprop$ it means that $T: \sortpropPos$ or $T: \sortpropNeg$. 
We say that a constructor $T$ is \emph{well-formed} if there exists some kind $\sigma$ such that $T: \sigma$ holds.

\begin{figure}[ht]   \begin{bigbox}
\begin{center}
\begin{tabular}{cc} 
\prooftree
 \justifies x^\tau : \tau
\thickness=0.08em
\shiftright 0em\using \endprooftree
&
\prooftree
 \justifies (x^\tau)^\bot : \tau^{\bot}
\thickness=0.08em
\shiftright 0em\using \endprooftree
\\ & \\

\prooftree
T : \tau \justifies \lambda x^\sigma.T : \sigma\rightarrow \tau
\thickness=0.08em
\shiftright 0em\using  \endprooftree
&
\prooftree
T : \sigma\rightarrow \tau \quad\quad U : \sigma \justifies TU : \tau
\thickness=0.08em
\shiftright 0em\using  \endprooftree\\& \\

\prooftree
 \justifies \foZero : \sortint
\thickness=0.08em
\shiftright 0em\using \endprooftree
&
\prooftree
\justifies \foSucc : \sortint \rightarrow \sortint
\thickness=0.08em
\shiftright 0em\using  \endprooftree\\& \\
\multicolumn{2}{c}{
\prooftree
 \justifies rec_\tau : \tau \rightarrow (\sortint \rightarrow \tau \rightarrow \tau) \rightarrow \sortint \rightarrow \tau
\thickness=0.08em
\shiftright 0em\using \endprooftree}\\ & \\

\multicolumn{2}{c}{
\prooftree
 \justifies rec_\tau^\bot :  \tau \rightarrow (\sortint \rightarrow \tau \rightarrow \tau) \rightarrow \sortint \rightarrow \tau
\thickness=0.08em
\shiftright 0em\using \endprooftree} \\
\end{tabular}
\end{center}

\begin{center}
\begin{tabular}{cccc} 
\prooftree
 A : \sortprop \quad\quad B : \sortprop \justifies A\otimes B : \sortpropPos
\thickness=0.08em
\shiftright 0em\using \endprooftree
&
\prooftree
 A : \sortprop \quad\quad B : \sortprop \justifies A\parr B : \sortpropNeg
\thickness=0.08em
\shiftright 0em\using \endprooftree& 

\prooftree
 A : \sortprop  \justifies \downarrow A : \sortpropPos
\thickness=0.08em
\shiftright 0em\using \endprooftree
& 
\prooftree
 A : \sortprop  \justifies \uparrow A : \sortpropNeg
\thickness=0.08em
\shiftright 0em\using \endprooftree 

\\  & \\
\multicolumn{2}{c}{
\prooftree
 A : \sortpropW \quad * \in \{+,-\}\justifies \forall x^\tau.A : \sortpropW
\thickness=0.08em
\shiftright 0em\using  \endprooftree} &

\multicolumn{2}{c}{\prooftree
 A : \sortpropW \quad * \in \{+,-\} \justifies \exists x^\tau.A : \sortpropW
\thickness=0.08em
\shiftright 0em\using \endprooftree}
\\ 
\end{tabular}
\end{center}\end{bigbox}
\caption{Typing rules for constructors}
\label{fig:HOMALkinds}
\end{figure}

\begin{prop}
If $T : \sigma$ then $T^{\bot} : \sigma^{\bot}$.
\end{prop}

Finally, we define a relation of \textit{convertibility} between constructors, noted $T \congE T'$, whose 
inductive definition is given in Figure \ref{fig:HOMALcongruence}. Notice that if $A$ and $B$ are
formulas and $A \congE B$ then $A$ and $B$ have the same polarity.
The presence of the dual recursor
$rec_\tau^{\bot}$ and its associated conversion rules are necessary to obtain the following property.

\begin{prop}
If $T$ and $U$ are constructors such that $T \congE U$, then
$T^\bot \congE U^\bot$. 
\end{prop}

\begin{figure}[ht!]\begin{bigbox}
\begin{center}
\begin{tabular}{cc}
 & \\
\prooftree
 \justifies (\lambda x^\tau.T)(U) \congE T\{ x^\tau := U \}
\thickness=0.08em
\shiftright 0em\using \endprooftree  &
\prooftree
x^\tau \notin FV(T) \justifies \lambda x^\tau.Tx \congE T
\thickness=0.08em
\shiftright 0em\using  \endprooftree\\ & \\
\prooftree
 \justifies rec_\tau \, T\, U\, 0\congE T
\thickness=0.08em
\shiftright 0em\using \endprooftree
& 

\prooftree
\justifies rec_\tau \, T\, U\, (s\, n)\congE U\, n\, (rec_\tau \, T\, U\, n)
\thickness=0.08em
\shiftright 0em\using  \endprooftree\\
& \\

\prooftree
 \justifies rec_\tau^\bot \, T\, U\, 0\congE T^\bot
\thickness=0.08em
\shiftright 0em\using \endprooftree
& 

\prooftree
\justifies rec_\tau^\bot \, T\, U\, (s\, n)\congE U^\bot\, n\, (rec_\tau \, T\, U\, n)
\thickness=0.08em
\shiftright 0em\using  \endprooftree\\
& \\
\end{tabular}
\end{center}

\begin{center}
\begin{tabular}{ccc}
\prooftree
 \justifies T \congE T
\thickness=0.08em
\shiftright 0em\using \endprooftree
& 
\prooftree
 T \congE T'\justifies T'\congE T
\thickness=0.08em
\shiftright 0em\using \endprooftree
&
\prooftree
T \congE T' \quad\quad T' \congE T''\justifies T \congE T''
\thickness=0.08em
\shiftright 0em\using  \endprooftree\\
&\\
\end{tabular}
\end{center}
\begin{center}
\begin{tabular}{cc}
\prooftree
 T \congE T'\justifies \lambda x^\tau.T \congE \lambda x^\tau. T'
\thickness=0.08em
\shiftright 0em\using \endprooftree
&
\prooftree
T\congE T'\quad\quad U\congE U'\justifies TU \congE T'U'
\thickness=0.08em
\shiftright 0em\using  \endprooftree\\
&\\

\prooftree
 A \congE A'\quad\quad B\congE B' \justifies A\otimes B \congE A'\otimes B'
\thickness=0.08em
\shiftright 0em\using \endprooftree
& \prooftree
 A \congE A'\quad\quad B\congE B' \justifies A\parr B \congE A'\parr B'
\thickness=0.08em
\shiftright 0em\using \endprooftree \\ &\\
\prooftree
A\congE A' \justifies\exists x^\tau\, A\congE \exists x^\tau\, A'
\thickness=0.08em
\shiftright 0em\using  \endprooftree
&
\prooftree
A\congE A' \justifies\forall x^\tau\, A\congE \forall x^\tau\, A'
\thickness=0.08em
\shiftright 0em\using  \endprooftree\\
& \\

\end{tabular}
\end{center}
\end{bigbox}\caption{Convertibility relation between constructors}
\label{fig:HOMALcongruence}
\end{figure}

\begin{rmk}
On the formulas constructor (of kind $\sortprop^{*}$ for $*\in \{+,-\}$), the negation $(.)^{\bot}$ is the usual involutive negation of linear logic. However, on closed individuals, it is simply the identity (modulo $\congE$). For example $((\lambda x^{\iota}. \foSucc \,x^{\iota}) \foZero)^{\bot} \congE \foSucc \,\foZero$.
\end{rmk}
  
Let us give a few examples of useful constructors we can define. 
\begin{itemize}
  \item The negation operator on positive formulas can be defined as the
  constructor $\lambda x^{\sortpropPos}. (x^{\sortpropPos})^{\bot}
  : \sortpropPos \rightarrow \sortpropNeg$. 
  Notice that the dual variable $(x^{\sortpropPos})^{\bot}$ is bound
  by the lambda binder $\lambda x^{\sortprop}$. 
  \item If we define $U = \lambda z^{\sortpropPos}.rec_\sortprop \,z^{\sortpropPos}\, (\lambda x^{\sortpropPos}. \lambda y^{\iota}.\downarrow (x^{\sortpropPos})^{\bot})$, we have
  \begin{eqnarray*}
  	  U \, A\, (\foSucc\, n) & \congE & \downarrow (U\, A\, n)^{\bot}
  \end{eqnarray*}
  For example, 
  \begin{eqnarray*}
  	  U\, A\, \mathbf{5} & \congE &  \downarrow \uparrow \downarrow\uparrow \downarrow A^{\bot}\\
  	  U\, A\, \mathbf{4} & \congE &  \downarrow \uparrow \downarrow \uparrow A\\
	 \end{eqnarray*}
  \end{itemize}

In the rest of the paper, we designate by the letters $N,M, \dots$ negative formulas 
and by the letters $P,Q,\dots$ positive formulas. We designate by the letters $A,B,\dots$ formulas
of any polarity (positive or negative). 
Negative formulas $N$ are intended to type negative terms (lazy terms), whereas positive formulas $P$ will be used to type positive terms (eager terms). The modality $\uparrow$ is used to turn a positive term into a negative one, that is transforms an eager term into a lazy one. $\downarrow$ does just the contrary, that is turning a negative term into a positive one. 

\begin{rmk}
In contrast with \cite{MunchReal}, $\forall$ and $\exists$ do not change the polarity
of the formula. This choice is made to keep realizers of existential and universal
statements simpler, especially when we will define forcing in section \ref{sec:forcing}.
\end{rmk}

\subsection{Type system} \label{subsec:HOMAL_type_system}

The type system \HOMAL relates terms
of {\LFOC } with \HOMAL formulas. 
\textit{Typing contexts} (denoted by the symbols: $\Gamma, \Gamma',\Delta, \dots$)
are finite sets containing elements of the form $x : N$ or $\alpha : P$.
\textit{Typing judgments} are of the form:
$$\vdash t_+ : P \delb \Gamma \quad\text{or}\quad \vdash t_- : N\delb \Gamma\quad\text{or}\quad c : (\vdash \Gamma)$$
The rules of \HOMAL are described in Figure \ref{fig:HOMALtyping}. 
Notice that in \class{MAL}$\omega$, only \emph{affine} terms are typable. That means that every bound variable $\kappa$ appears at most once
in the command under the binder. 

\begin{rmks}
$ $

\begin{enumerate}
\item
We have chosen to have weakening and conversion rules expressed only on commands.
The reason is that weakening and conversion rules for terms are derivable from these two 
rules, using the cut and activation rules. For example, here is the derived rule 
 $(\congE)$ on terms (the case of weakening is similar):
$$
\prooftree
	\prooftree
			\prooftree
				 \vdash t : A \delb \Gamma \quad\quad
				   	\prooftree \justifies \vdash \kappa : A^{\bot} \delb \kappa : A
				 		\thickness=0.08em
	    				\shiftright 0em\using (Ax)
					 \endprooftree	
	   			 \justifies \com{t}{\kappa} : (\vdash A, \Gamma) 
				 \thickness=0.08em
	    		\shiftright 0em\using  (Cut)
			\endprooftree	
		   \quad \quad A \congE B
	    \justifies \com{t}{\kappa} : (\vdash B, \Gamma)
  	    \thickness=0.08em
	    \shiftright 0em\using  (\congE)
	\endprooftree  \justifies \vdash \mu \kappa.\com{t}{\kappa} : B \delb \Gamma
  \thickness=0.08em
  \shiftright 0em\using (\mu)
\endprooftree
$$
\item We can only form pairs of values $(V,V')$. This is reflected
in the type system by the $(\otimes)$ rule that introduces only
such pairs. However, we can obtain a derived rule for the
following definition of $(t,t')$ already presented in Subsection \ref{subsec:reduction}:
$$(t,t') = \mu \alpha.\com{t}{\mu x.\com{u}{\mu y.\com{\alpha}{(x,y)}}}$$
We just give the partial derivation corresponding to the derived rule, leaving
the easy part to the reader.
$$
\prooftree
	\prooftree
	 			\mathbf{\vdash t : A \delb \Gamma}
				\prooftree 
						\prooftree
	 								\mathbf{\vdash u : B \delb \Delta}
									\quad
									\prooftree
											\prooftree
													\vdots
													\justifies \com{(x,y)}{\alpha} : 	
															(\vdash x :A^{\bot}, y : B^{\bot}, \alpha: A \otimes B)
				 							\thickness=0.08em
	    									\shiftright 0em\using (Cut)
					 						\endprooftree
											\justifies \mu y.\com{(x,y)}{\alpha} : B^{\bot} \delb
														 x: A^{\bot},\alpha : A\otimes B
				 					\thickness=0.08em
	    							\shiftright 0em\using (\mu)
					 				\endprooftree												
								\justifies \com{u}{\mu y.\com{(x,y)}{\alpha}} :
										 (\vdash x : A^{\bot}, \alpha : A\otimes B, \Delta)
				 		\thickness=0.08em
	    				\shiftright 0em\using (Cut)
					 	\endprooftree	
						\justifies
							  \vdash \mu x.\com{u}{\mu y.\com{(x,y)}{\alpha}} : A^{\bot} \delb 
							  		\Delta, \alpha : A\otimes B 
				 		\thickness=0.08em
	    				\shiftright 0em\using (\mu)
					 \endprooftree	
	    \justifies \com{t}{\mu x.\com{u}{\mu y.\com{(x,y)}{\alpha}}} : (\vdash \alpha: A\otimes B, \Gamma,\Delta)
  	    \thickness=0.08em
	    \shiftright 0em\using  (Cut)
	\endprooftree  
	  \justifies\mathbf{\vdash (t,u) : A\otimes B \delb \Gamma,\Delta}
  \thickness=0.08em
  \shiftright 0em\using (\mu)
\endprooftree
$$
The same remark also holds for the construction $\{t\}$ defined in Subsection \ref{subsec:reduction}.
\end{enumerate}
\end{rmks}

\begin{figure}[ht]
\begin{bigbox}
\begin{center}
\begin{tabular}{cc}
& \\
\prooftree
   \justifies \vdash x: P \delb x : P^\bot
\thickness=0.08em
\shiftright 0em\using (Ax_+) \endprooftree
&
\prooftree
   \justifies \vdash \alpha : N \delb \alpha : N^\bot
\thickness=0.08em
\shiftright 0em\using (Ax_-) \endprooftree\\
&  \\
\prooftree
   c : (\vdash \alpha : P, \Gamma) \justifies \vdash \mu \alpha.c : P \delb \Gamma
\thickness=0.08em
\shiftright 0em\using (\mu_+) \endprooftree
&

\prooftree
   c : (\vdash x : N, \Gamma) \justifies \vdash \mu x.c : N \delb \Gamma
\thickness=0.08em
\shiftright 0em\using (\mu_-) \endprooftree
\\
& \\
\multicolumn{2}{c}{
\prooftree
   \vdash t : A \delb \Gamma \quad\quad \vdash u : A^\bot \delb \Delta \justifies \com{t}{u} : (\vdash \Gamma, \Delta)
\thickness=0.08em
\shiftright 0em\using (Cut) \endprooftree}\\
& \\
\end{tabular}
\end{center}

\begin{center}
\begin{tabular}{cc}
\prooftree
  \vdash V : A\delb\Gamma \quad\quad \vdash  V' : B \delb \Delta \justifies \vdash (V,V'): A\otimes B \delb \Gamma,\Delta
\thickness=0.08em
\shiftright 0em\using (\otimes) \endprooftree
&
\prooftree
   c: (\vdash \kappa : A,\kappa':B,\Gamma) \justifies \vdash \mu (\kappa,\kappa').c : A\parr B\delb \Gamma
\thickness=0.08em
\shiftright 0em\using (\parr) \endprooftree\\
&  \\
\prooftree
   \vdash V: A \delb \Gamma\justifies \vdash \{ V \} : \downarrow A \delb \Gamma
\thickness=0.08em
\shiftright 0em\using (\downarrow) \endprooftree
&
\prooftree
   c: (\vdash \kappa : A , \Gamma) \justifies \vdash \mu \{ \kappa \}.c : \uparrow  A \delb \Gamma
\thickness=0.08em
\shiftright 0em\using (\uparrow) \endprooftree
\\
& \\
\prooftree
    T : \tau\quad  \vdash t : A[T/x^\tau]\delb \Gamma\justifies \vdash t: \exists x^\tau.A \delb \Gamma
\thickness=0.08em
\shiftright 0em\using (\exists) \endprooftree
&
\prooftree
   \vdash V : A \delb \Gamma \quad x\text{ does not appear in }\Gamma \justifies \vdash V : \forall x^\tau.A \delb\Gamma
\thickness=0.08em
\shiftright 0em\using (\forall) \endprooftree\\
& \\ 
\end{tabular}
\end{center}

\begin{center}
\begin{tabular}{cc}

\prooftree
  c: (\vdash \kappa: A, \Gamma) \quad\quad A\congE B\justifies c: (\vdash \kappa : B , \Gamma)
\thickness=0.08em
\shiftright 0em\using (\congE) \endprooftree &
\prooftree
  c: (\vdash \Gamma) \justifies c: (\vdash x : A , \Gamma)
\thickness=0.08em
\shiftright 0em\using (W) \endprooftree \\
\end{tabular}
\end{center}
\end{bigbox}
\caption{Typing rules of \HOMAL}

\label{fig:HOMALtyping}
\end{figure}

\subsection{Generalities on Call-by-value and Call-by-name}\label{subsec:cbncbv}

In this focalising version of \class{MAL}$\omega$, it is possible to encode both call-by-name and 
call-by-value affine $\lambda$-calculi, as explained in \cite{MunchReal}. Moreover, each $\beta$-reduction step in these 
calculi induces a constant number of reduction steps in the corresponding encoding.\newline

\subtitle{Call-by-name $\lambda$-calculus} We consider the call-by-name affine $\lambda$-calculus, that is
such that in every term $\lambda x.t$, $x$ appears at most once in $t$. 
We exhibit an encoding of this calculus by giving a typed translation of the affine $\lambda$-calculus in the negative fragment of \HOMAL. The implication $\multimap$ is defined as follows:
	$$N\multimap M  \quad\defequ\quad   N^\bot \parr M$$
Terms $t$ and stacks $\pi$ of the Krivine abstract machine \cite{krivine2007call} are encoded respectively using 
negative terms $\comr{t}$ and positive terms $\coml{\pi}$, as follows:
\begin{eqnarray*}
	{\lambda \alpha.t_-} & \defequ & {\mu(\alpha,x).\com{t_-}{x}}\\
  {u.\pi} &  \defequ & {(u,\pi)}\\
  {(t_-)u_-} & \defequ & {\mu x.\com{t_-}{u_-.x}}\\
\end{eqnarray*}
We can check that these definitions indeed implement
Krivine Machine weak call-by-name reduction, as shown by
the following reduction:

\begin{eqnarray*}
  \com{(\lambda \alpha.t)u}{\pi} &  = & \com{\mu x.\com{\mu (\alpha, y).\com{t}{y}}{u.x}}{\pi}\\
   & \redh & \com{\mu(\alpha,y).\com{t}{y}}{(u,\pi)}\\
   & \redh & \com{t[u/\alpha]}{\pi}\\
\end{eqnarray*}

We see that each $\beta$-reduction step in the weak call-by-name $\lambda$-calculus
corresponds exactly to two $\redh$ reduction steps in this encoding. \newline

\subtitle{Call-by-value $\lambda$-calculus} The call-by-value affine $\lambda$-calculus
is obtained by taking a positive encoding of the implication:
$$P\multimap Q  \quad\defequ\quad \downarrow (P^\bot \parr Q)$$
We define terms and environments using respectively positive and negative terms:
\begin{eqnarray*}
	{\lambda x.t_+} & \defequ & {\{\mu(x,\alpha).\com{t_+}{\alpha}\}}\\
  {(t_+)u_+} & \defequ & {\mu \alpha.\com{t_+}{u_+.\alpha}}\\
  {u.e} &  \defequ & {\mu \{\alpha\}.\com{\alpha}{(u,e)}}\\
\end{eqnarray*}
It can be checked that we retrieve Curien-Herbelin $\bar{\lambda}\mu\tilde{\mu}_v$ calculus \cite{curien2000duality}. 
Here is the typical example of a reduction in the encoded calculus:

\begin{eqnarray*}
  \com{(\lambda x.t)u}{E} &  = & \com{\mu \alpha.\com{\{\mu(x,\alpha').\com{t}{\alpha'}\}}{u.\alpha}}{E}\\
   & \redh & \com{\{\mu(x,\alpha').\com{t}{\alpha'}\}}{\mu \{\alpha\}.\com{\alpha}{(u,E)}}\\
   & \redh & \com{\mu(x,\alpha').\com{t}{\alpha'}}{(u,E)}  \\
   & \redh & \com{u}{\mu \kappa_1.\com{E}{\mu \kappa_2.\com{(\kappa_1,\kappa_2)}{\mu (x,\kappa).\com{t}{\kappa}}}}\\
   & \redh^* & \com{V}{\mu \kappa_1.\com{E}{\mu \kappa_2.\com{(\kappa_1,\kappa_2)}{\mu (x,\kappa).\com{t}{\kappa}}}}\\
   & \redh & \com{E}{\mu\kappa_2.\com{(V, \kappa_2)}{\mu (x,\kappa).\com{t}{\kappa}}}\\
   & \redh & \com{(V, E)}{\mu (x,\kappa).\com{t}{\kappa}}\\
   & \redh & \com{t[V/x]}{E}\\
\end{eqnarray*}

Here again, it is clear that to each step in the Curien-Herbelin $\bar{\lambda}\mu\tilde{\mu}_v$ calculus
corresponds a constant number of steps in \LFOC.\newline

\section{Quantitative Krivine's realizability}\label{sec:QuantKrivine}

In this section, we define the quantitative classical realizability for \class{MAL}$\omega$.
This construction is a direct extension of Munch's focalised version of Krivine's
classical realizability \cite{MunchReal} and integrates the quantitative aspects 
of \cite{DalLago20112029}. 
In (non-quantitative) Krivine's classical realizability, formulas are interpreted 
as sets of terms closed by a notion of biorthogonality, and a realizability relation 
$t \Vdash A$ between terms and formulas is defined. In that setting, $t\Vdash A$ intuitively
means \emph{``$t$ is a term whose computational behavior follows the specification $A$''}.
In our work, we interpret formulas
$A$ as sets of pairs $(t,p)$ where $t$ is a term and $p$ is an abstract 
quantity. The realizability relation becomes $(t,p) \Vdash A$, with the
informal meaning \emph{``$t$ is a term whose computational behavior follows
the specification $A$ and uses during its execution a quantity of resources bounded by $p$''}. 
The presence of this abstract quantity allows us to build
a quantitative extension of the well-known technique of reducibility candidates, which
we call \emph{quantitative reducibility candidates}. We
use these to prove bounded-time termination results on typable terms.

\subsection{Quantitative monoid}

We introduce the notion of \textit{quantitative monoid}, whose elements can be thought as 
resource quantities (like time, space or energy). Quantitative monoids are a generalization 
and a simplification of Dal Lago and Hofmann's resource monoids \cite{DalLago20112029}.

\begin{defn} A \emph{quantitative monoid} is
a structure $(\M, +, \zero, \leq, \norm{.})$ where:
\begin{itemize}
  \item $(\M, +, \zero, \leq)$ is a preordered commutative monoid.
  \item $\norm{.} : \M \longrightarrow \N$ is a function
  such that:
   \begin{itemize}
   \item for every $p,q \in \M$, we have $\norm{p}+\norm{q} \leq \norm{p+q}$. 
  \item  Morever, $\norm{.}$ is compatible with $\leq$, that is
if $p \leq q$ then $\norm{p} \leq \norm{q}$.
  \end{itemize}
\end{itemize}
If moreover, there is an element $\unit\in \M$  such that $1 \leq \norm{\unit}$,
then we say that $(\M, +, \zero, \unit, \leq, \norm{.})$ is a \emph{quantitative
monoid with unit}.
\end{defn}

From now on, we will often denote a quantitative monoid
by its carrier $\M$  and 
we use lower-case consonnes letters $p,q,m,v,\dots$ to denote 
its elements. Moreover, if $n\in\N$ then we use the notation $n.p$ for $\underbrace{p+p+\dots + p}_{n\text{ times}}$.

\begin{rmks}
$ $

\begin{enumerate}
\item
If we think of elements $p,q\in \M$ as abstract quantities bounding respectively the resource consumption
of programs $t$ and $u$, then doing the operation $p+q$ can be seen as way to calculate 
a bound for the resource consumption of the process $\com{t}{u}$ resulting
of the interaction of these two programs. 
\item
The intuition behind the anti-triangular inequality,
$\norm{p} + \norm{q} \leq \norm{p+q}$
 is that 
the amount of resources potentially used by the interaction of two 
programs is more than the sum of the quantities of
resources used by the two programs alone.
\item
One corollary of the anti-triangular inequality is that 
$\norm{\zero} = 0$. Indeed,
$$2\times\norm{\zero} = \norm{\zero}+\norm{\zero} \leq \norm{\zero}$$

\end{enumerate}
\end{rmks}

\begin{exa}\label{ex:int_monoid}
The structure $(\N, +, 0, 1, \leq, x\mapsto x)$ where $+$ is the usual addition on integers
and $\leq$ is the usual order on $\N$, is a quantitative monoid with unit. 
\end{exa}

\begin{rmk}
In any quantitative monoid with unit $\M$ we have elements of arbitrary big measure, that is
for every $n\in \N$,
$$n \leq \norm{n.\unit}$$
\end{rmk}

It has to be noted that every resource monoid in the sense of \cite{DalLago20112029} defines a 
quantitative monoid with unit, by choosing $\norm{p} = \mathcal{D}(\zero,p)$ where $\mathcal{D}(.,.)$ is
the distance of the resource monoid.

%


\subsection{Quantitative pole and orthogonality} 

Krivine's classical realizability is a framework parametrized by a set $\Bot$
of commands called the \emph{pole}. This set can be seen as the set of
\emph{correct} processes, that is the notion of correctness we want to study. 
For example, $\Bot$ can be the set of normalizing commands. 
This set then induces a notion of orthogonality, which is used to define
an interpretation of the type system.

Similarly, our model of \HOMAL will be parametrized by a structure called
\emph{quantitative pole} that we defined now.

\begin{defn}
Let $\M$ be a quantitative monoid. We define the notions of \emph{weighted terms}
and \emph{weighted commands} as follows:
\begin{enumerate}
  \item A \emph{weighted term} is a pair $(t,p)\in \Terms^{0}\times \M$.
  \item A \emph{weighted command} is a pair $(c,p)\in \ComSet \times \M$.
\end{enumerate}
\end{defn}

Informally, a weighted command $(c,p)$ carries a quantitative information $p$ that
is a bound on the amount of resources used by $c$ during its execution. 
We also define the notion of \emph{quantitative pole}, which will be the main
parameter of our model.

\begin{defn}[Quantitative pole]
A \emph{quantitative pole} is a pair $(\M, \Bot)$ where:
\begin{itemize}
  \item $\M$ is a quantitative monoid. 
  \item $\Bot \subseteq \ComSet\times\M$ is a set of weighted
  commands. 
\end{itemize}
When it is clear from the context, we will often refer to a quantitative pole using its
set $\Bot$. 
\end{defn}

As we will see, not all quantitative poles yield sound interpretations of \HOMAL.
We define a subclass of quantitative poles, the \emph{saturated quantitative poles}.

\begin{defn}[Saturated pole]
A \emph{saturated quantitative pole} is a structure \\
$(\M, \Bot, \pBeta)$ given by:
  \begin{itemize}
    \item A quantitative pole $(\M, \Bot)$. 
    \item An element $\pBeta$ of $\M$.
    \item Moreover $\Bot$ satisfies 
    the following properties:  
    
           \begin{tabular}{rl}
   			 {\bf ($\rightarrow_\beta\text{-saturation}$)} & If $c \rightarrow_\beta c'$ and $(c',p) \in \Bot$
			 	then $(c,p+\pBeta) \in \Bot$ \\
   			 {\bf ($\rightarrow_\mu\text{-saturation}$)} & If $c \rightarrow_\mu c'$ then $(c',p) \in \Bot \Longleftrightarrow (c,p) \in \Bot$ \\
 			 {\bf ($\leq$-saturation) }& If ${p}\leq {p'}$ and $(c,p)\in \Bot$ then $(c,p')\in \Bot$. 
		 \end{tabular}  
	 \end{itemize}
\end{defn}

\begin{rmk}
The element $\pBeta$ corresponds informally to the 
\textit{cost} of a single $\beta$-reduction step, as witnessed by the $\redBeta$-saturation
property. The $\redMu$ steps, however, are not considered as a resource cost, since they 
are mainly administrative reductions. 
\end{rmk}

\begin{exa}\label{ex:quantpole}
\begin{itemize}
\item 
Suppose $P$ is a \textit{non-quantitative saturated pole}, i.e a set of commands which is closed under anti-evaluation
(i.e: $c\redh c'\wedge c'\in  P \Rightarrow c\in P$).
If $\M$ is a quantitative monoid and $\pBeta \in \M$, then
$(\M, P \times\M, \pBeta)$ is a saturated quantitative pole.

\item
An important example is $\BotTime$ the set of bounded time terminating
processes, namely $\BotTime = \setp{ (c,p)}{\Time(c) \leq \norm{p}}$.
In particular all $(c,p)\in\BotTime$ are such that $c$ terminates. If $\M$
has a unit $\unit$, then
$\BotTime$ provides a saturated quantitative pole by choosing $\pBeta =\unit$.
The $\redBeta$-saturation property relies on the fact that if 
$c\redBeta c'$ and $(c',p)\in\BotTime$, then $$\Time(c)  = \Time(c') + 1 \leq \norm{p} + 1 \leq \norm{p} + \norm{\unit}\leq \norm{p+\unit}$$
\end{itemize}
\end{exa}

Until the rest of this section, 
we assume a choice of a quantitative pole $(\M, \Bot)$ (which is not necessarily saturated). 
This quantitative pole $\Bot$ induces a notion of \emph{orthogonality} between 
elements of $\PosTerms\times\M$ and $\NegTerms\times\M$.

\begin{defn}[Orthogonality]We say that $(t_+,p)\in \PosTerms\times\M$ and $(t_-,q)\in \NegTerms\times\M$ are \textit{orthogonal} and we note:
$$(t_+,p)\bot (t_-,q) \Longleftrightarrow (\com{t_+}{t_-},p+q) \in \Bot$$
This orthogonality relation is lifted as an operation on set of bounded terms. 
If $X \subseteq \PosTerms\times\M$, then we define its orthogonal as
$$X^\bot \equiv \setp{(t_-,q)}{\forall (t_+,p)\in X, t_+ \bot t_-}$$
In a similar way, if $X \subseteq \NegTerms\times\M$, then
$$X^\bot \equiv \setp{(t_+,p)}{\forall (t_-,q)\in X, t_+ \bot t_-}$$
\end{defn}

\begin{rmk}
Informally, the meaning of $(t_+,p) \bot (t_,q)$ is that 
the interaction $\com{t_+}{t_-}$ \textit{behaves well}
and uses an amount of resources \textit{bounded} by $p+q$. 
\end{rmk}

The operation $(.)^\bot$ satisfies the usual properties of orthogonality. 

\begin{prop}
If
$X$ and $Y$ are both subsets of $\PosTerms\times\M$ (resp. subsets of $\NegTerms\times\M$) then we have:
\begin{itemize}
  \item $X \subseteq X^\dbot$
  \item $X \subseteq Y$ implies $Y^\bot \subseteq X^\bot$
  \item $X^{\bot\bot\bot} = X^\bot$
\end{itemize}
\end{prop}

\begin{prop}\label{prop:ortho_cap}
If $(X_i)_{i\in I}$ is a family of subsets of $\PosTerms\times\M$ (resp. 
$\NegTerms\times \M$), then the following equalities hold:
\begin{enumerate}
  \item $(\bigcup_{i\in I} X_i)^{\bot} = \bigcap_{i\in I} X_i^{\bot}$
  \item $(\bigcap_{i\in I} X_i)^{\bot} = (\bigcup_{i\in I} X_i^{\bot})^{\bot\bot}$
\end{enumerate}
\end{prop}
Finally, we define a few notations. Let $X\in \mathcal{P}(\PosTerms\times\M)\cup
\mathcal{P}(\NegTerms\times\M)$. Then:
\begin{eqnarray*}
  \BV{X}        & = &  X \cap (\mathbb{V} \times \M)\\
  \overline{X} & = &  \setp{(t,q)}{\exists p \leq q\text{ such that } (t,p)\in X}
\end{eqnarray*}

\begin{rmk}
If $X \in \mathcal{P}(\NegTerms\times\M)$ then $\BV{X} = X$.
Indeed, every negative term is also a value. 
\end{rmk}

\subsection{Interpretation of kinds} 
Before giving the actual interpretation of kinds and constructors,
we define two operations on sets of bounded terms. 
\begin{eqnarray*}
  X\otimes Y & = & \setp{((t,u), p+q)}{(t,p)\in X\wedge (u,q)\in Y}\\
  \downarrow X & = & \setp{ (\{t\}, p) }{(t,p)\in X}
\end{eqnarray*}

Suppose $\Dpos \in \mathcal{P}(\PosTerms\cap \mathbb{V}\times\M)$. We then define 
$\Dneg$ as the set
$\setp{X^{\bot}}{X\in \Dpos}$ and we pose  $\mathcal{D} = \Dpos \cup \Dneg$.

\begin{defn}[Propositional domain] \label{def:prop_domain}
We say that $\Dpos$ is a \emph{positive propositional domain} if it satisfies the following properties: 
\begin{itemize}
  \item If $(X_i)_{i\in I}$ is a family of elements of $\Dpos$
  indexed by $I$, then $\bigcup_{i\in I} X_i \in \Dpos$ and 
  $\bigcap_{i\in I} X_i \in \Dpos$. 
  \item If $X,Y \in \mathcal{D}$, then $X \otimes Y \in \Dpos$.
  \item If $X \in \mathcal{D}$, then $\downarrow X \in \Dpos$.
\end{itemize}
\end{defn}

Now suppose we have fixed a positive propositional domain $\Dpos$. 
We begin with the interpretation of kinds.  If $\sigma$ is a kind we define 
its \textit{interpretation} $\Inte{\sigma}$:
\begin{center}
\begin{tabular}{rcl}
	$\Inte{\sortint}  $					      & $=$ & $\N$ 			  \\ 
	$\Inte{\sortpropPos}  $					  & $=$ & $\Dpos$ \\
	$\Inte{\sortpropNeg}  $					  & $=$ & $\Dneg$ \\
	$\Inte{\sigma\rightarrow\tau}$   & $=$ & $\Inte{\tau}^{\Inte{\sigma}}$\\
\end{tabular}
\end{center}

\subsection{Interpretation of constructors}\label{subsec:interpcons}
The orthogonality operation $(.)^{\bot}$ defined earlier on
set of bounded terms, is extended inductively on elements of all kinds.
That is, for $T$ being an element of $\Inte{\sigma}$, we define $\bot(T,\sigma)$ 
as:
\begin{eqnarray*}
  \bot(X, \iota) & = & X\\
  \bot(X, \sortpropPos) & = & X^{\bot}\\
  \bot(X^{\bot}, \sortpropNeg) & = & X\\
  \bot(X, \sigma \rightarrow \tau) & = & Y \in \Inte{\sigma} \mapsto \bot(X(Y), \tau)
 \end{eqnarray*}
  Notice that this definition makes sense only because $\Dpos$ is a positive propositional
  domain. Hence we know that any element of $\Inte{\sortpropNeg}$ is the orthogonal
  of an element of $\Inte{\sortpropPos}$. 
 \begin{exa}
 If $X \in \Dpos$, then $\bot(X, \sortpropPos)$ coincide with the orthogonal $X^{\bot}$
 of $X$.  On the kind $\sortpropPos\rightarrow\sortpropPos$, consider for example
 $X\mapsto X \in \Inte{\sortpropPos\rightarrow\sortpropPos}$, we obtain
 $\bot(X\mapsto X, \sortpropPos\rightarrow\sortpropPos) = X\mapsto X^{\bot} \in 
 \Inte{\sortpropPos \rightarrow \sortpropNeg}$,
 that is the orthogonality operator. 
 \end{exa} 
 
 This notion of extended orthogonality is consistent with the syntactic orthogonality
 on kinds, as witnessed by the following property. 
 \begin{prop}
 If $T \in \Inte{\sigma}$, then we have $\bot(T,\sigma) \in \Inte{\sigma^{\bot}}$. 
 \end{prop}
 \begin{proof}
 It is proved by induction on the kind $\sigma$, and is a consequence
 of the definition of $\sigma^{\bot}$ and the fact that $\Dpos$ 
 is a propositional domain. 
 \end{proof}
 
 Given the positive propositional domain, a \emph{valuation} is a partial function $\rho$ 
 assigning to a variable $x^\sigma$ of kind $\sigma$ an element $\rho(x^\sigma) 
 \in \Inte{\sigma}$. We denote by $\rho[x^{\sigma}\leftarrow v]$ the valuation
 obtained from $\rho$ by (re)binding the variable $x^{\sigma}$ to the element
 $v\in \Inte{\sigma}$. We say that $\rho$ \emph{closes} a constructor $T$
 if $FV(T) \subseteq dom(\rho)$ and we note it $\rho \Vdash T$. By extension,
 we denote by $\rho \Vdash T_1,\dots, T_n$ if $\rho$ closes each constructor 
 $T_i$. A \emph{total valuation} is a valuation whose domain is the set of 
 all higher-order variables. If $\rho$ is a total valuation, then for every constructor
 $T$, $\rho \Vdash T$.  
 
 Given 
a well-typed constructor $T$ and a valuation $\rho$ such that
$\rho\Vdash T$, we define the set 
$\Inte{T}_\rho$ by induction on $T$:
\begin{center}
\begin{tabular}{rcl}
	$\Inte{x^\sigma}_\rho$					  & $=$ & $\rho (x)$ 													 \\
	$\Inte{(x^\sigma)^\bot}_\rho$ &   $=$ &  $\bot(\rho(x), \sigma)$  \\
	$\Inte{\lambda x^\sigma.T}_\rho$       & $=$ & $(v\in \Inte{\sigma}\mapsto \Inte{T}_{\rho[x\leftarrow v]})$\\
	$\Inte{TU}_\rho$					  & $=$ & $(\Inte{T}_\rho)\Inte{U}_\rho$ 													 \\
	$\Inte{\foZero}_\rho$ 			& $=$ & $0$\\
	$\Inte{\foSucc}_\rho$ 			& $=$ & $n \mapsto n+1$\\
	$\Inte{rec_\tau}_\rho$ 			& $=$ & $rec_\Inte{\tau}$\\
	$\Inte{rec^{\bot}_\tau}_\rho$ 			& $=$ & $\bot(rec_\Inte{\tau}, (\tau \rightarrow (\iota\rightarrow \tau \rightarrow \tau) \rightarrow \iota \rightarrow \tau))$\\
\end{tabular}
\end{center}

Concerning constructors $A$ which are formulas, that is of kind $\sortpropPos$
and $\sortpropNeg$, the set $\Inte{A}_\rho$ is an element of 
$\mathcal{P}(\PosTerms\times \M)\cup \mathcal{P}(\NegTerms\times\M)$. 
Moreover, it contains only values:
\begin{center}
\begin{tabular}{rcl}
  $\Inte{\downarrow A}_\rho$ & $=$ & $\downarrow{\Inte{A}_\rho}$ \\
  $\Inte{\uparrow A}_\rho$   & $=$ & $(\downarrow \Inte{A^\bot}_\rho)^\bot$\\
	$\Inte{A\otimes B}_\rho$   & $=$ & $\Inte{A}_\rho \otimes \Inte{B}_\rho$  \\
	$\Inte{A\parr B}_\rho$     & $=$ & $(\Inte{A^\bot}_\rho \otimes \Inte{B^\bot}_\rho)^\bot$\\
  $\Inte{\exists x^\sigma.P}_\rho$  & $=$ & $\bigcup_{v\in \Inte{\sigma}} {\Inte{P}_{\rho[x^{\sigma}\leftarrow v]}}$ \\
  $\Inte{\forall x^\sigma.N}_\rho$  & $=$ & $(\bigcup_{v\in \Inte{\sigma}} {\Inte{N^{\bot}}_{\rho[x^{\sigma}\leftarrow v]}})^\bot$ \\
  $\Inte{\forall x^\sigma.P}_\rho$  & $=$ & $\bigcap_{v\in \Inte{\sigma}} {\Inte{P}_{\rho[x^{\sigma}\leftarrow v]}}$ \\
  $\Inte{\exists x^\sigma.N}_\rho$  & $=$ & $(\bigcap_{v\in \Inte{\sigma}} {\Inte{N^{\bot}}_{\rho[x^{\sigma}\leftarrow v]}})^{\bot}$ \\
\end{tabular}
\end{center}

Finally, for a formula $A$, we define the set $\interp{A}_\rho$ as
$$\interp{A}_\rho \quad = \quad \Inte{A}_\rho^{\dbot}$$
\begin{rmks}\label{rmk:interp1}
$ $
\begin{enumerate}
\item For each well-typed constructor $T : \sigma$ and each 
valuation $\rho \Vdash T$, we have $\Inte{T}_\rho \in \Inte{\sigma}$. 
This rely on the fact that $\Dpos$ is a positive propositional domain, and hence
is closed under the required operations. 
\item For the negative existential case, we notice that
$$\Inte{\exists x^{\tau}.N}_\rho = (\bigcup_{v \in \Inte{\tau}} \Inte{N}_{\rho[x^{\tau}\leftarrow v]})^{\bot\bot}$$. 
\item However, for the universal case, the interpretation of $\forall x^{\tau}.A$ is always
$$\Inte{\forall x^{\tau} .A}_\rho =  \bigcap_{v\in \Inte{\tau}} {\Inte{A}_{\rho[x^{\tau}\leftarrow v]}}$$
even if $A$ is negative. We gave different formulations for the negative and positive cases in order
to show clearly that $\Inte{\forall x^{\tau}.N}_\rho = \Inte{\exists x^{\tau}.N^{\bot}}_\rho^{\bot}$, but
this remark shows it is not mandatory. 
\end{enumerate}
\end{rmks}
If $T$ is a closed well-typed constructor, $\Inte{T}_\rho$ and $\interp{T}_\rho$ are independent of
  $\rho$. Hence we will often simply note them respectively $\Inte{T}$ and $\interp{T}$.
  
\begin{prop} \label{prop:interp1}
The interpretation $\Inte{.}$ enjoys the following properties.   
  \begin{enumerate}
  	\item If  $T$ and $T'$ are two well-typed constructors such that $T\congE T'$, then
	for every valuation $\rho \Vdash T,T'$, we have
	 $T \congE T'$ then $\Inte{T}_\rho = \Inte{T'}_\rho$.
  	\item For any constructor $T$ of kind $\sigma$ and $\rho \Vdash T$, we have 
  $\Inte{T^{\bot}}_\rho = \bot (\Inte{T}_\rho, \sigma)$. 
  \item For any well-typed constructors $T: \tau$ and $S : \sigma$, any
  valuation $\rho$ such that $\rho \Vdash S$ and
  $FV(T)\subseteq dom(\rho)\cup \{ x^{\sigma} \}$, we have
  $$\Inte{T}_{\rho[x^{\sigma} \leftarrow \Inte{S}_\rho]} = \Inte{T[S/x^{\sigma}]}_\rho$$ 
  \end{enumerate}
\end{prop}
\begin{proof}
\begin{enumerate}
  \item This is immediate by induction first on the kind $\sigma$ and on the judgment $\congE$.
  \item This is proved by induction on the typing judgment of the constructor $T$. 
  \item This is proved  by induction on the typing judgment of the constructor $T$. 
\end{enumerate}
\end{proof}

\begin{rmk}\label{rmk:interp_sat}
Notice that neither the definition of the interpretation of well-typed constructors
nor the proof of Property \ref{prop:interp1} need to suppose that $\Bot$
is saturated.
\end{rmk}

We say that $(t,p)$ \textit{realizes} a closed formula $A$ and we note $(t,p) \Vdash^{\rho} A$
iff $(t,p) \in \Inte{A}_{\rho}$. If the choice of $\rho$ is clear from the context we only note
it $(t,p)\Vdash A$. 
 We may sometimes use the notation $(t,p)\Vdash_\Bot A$ to precise the quantitative pole we consider. 

\subsection{Properties of saturated quantitative poles}\label{subsec:sat_prop}

Until now, we have considered a quantitative pole which is not necessarily saturated.
When the pole is saturated, we can derive many properties that will be crucial
to prove that our model is sound with respect to \class{MAL}$\omega${.}
In this subsection, we suppose that $(\M, \Bot, \pBeta)$ is a saturated
quantitative pole, and explore the properties satisfied by the orthogonality operation. 
The first property we prove expresses the fact that a set closed by biorthogonality is also
$\leq$-saturated. 

\begin{prop}
For every $X\in \mathcal{P}(\PosTerms\times\M)\cup
\mathcal{P}(\NegTerms\times\M)$, we have 
$\overline{X} \subseteq X^{\dbot}$.
\end{prop} 
\begin{proof}
Let $(t,p)\in X$ and $q\in \M$ such that $p\leq q$. 
If $(u,r) \in X^{\bot}$ then $(t,p)\bot (u,r)$. By $\leq$-saturation 
of $\Bot$, we have $(t,q)\bot (u,r)$. Hence $(t,q)\in X^{\dbot}$. 
\end{proof}

The following lemma proves that we can safely remove or add double orthogonal operators
in interpretations of constructors. 
\begin{lem}\label{lem:gen_tensor_shift}
Suppose $X,Y \in \mathcal{P}(\PosTerms\times\M) \cup\mathcal{P}(\NegTerms\times\M)$
and $D \subseteq \mathcal{P}(\PosTerms\times \M) \cup 
\mathcal{P}(\NegTerms\times\M)$ with $D\neq\emptyset$. 
Then we have the following equalities:
\begin{enumerate}
  \item $(\downarrow X)^{\bot\bot} = (\downarrow X^{\bot\bot})^{\bot\bot}$
  \item $(X \otimes Y)^{\bot\bot} = (X^{\bot\bot}\otimes Y^{\bot\bot})^{\bot\bot}$
  \item $(\bigcap_{X\in D} X)^{\bot\bot} = \bigcap_{X\in D} X^{\bot\bot}$
  \item $(\bigcup_{X\in D} X)^{\bot\bot} = (\bigcup_{X\in D} X^{\bot\bot})^{\bot\bot}$
 \end{enumerate}
\end{lem}
\begin{proof}
  \begin{enumerate}
    \item Since $X \subseteq X^{\bot\bot}$, we immediately have 
    $(\downarrow X)^{\bot\bot} \subseteq (\downarrow X^{\bot\bot})^{\bot\bot}$. Let's
    prove that $(\downarrow X^{\bot\bot})^{\bot\bot} \subseteq (\downarrow X)^{\bot\bot}$.
    It suffices to show that $(\downarrow X)^{\bot}\subseteq (\downarrow X^{\bot\bot})^{\bot}$. 
    Let $(t,q)\in (\downarrow X)^{\bot}$ and $(u,p)\in  X^{\bot\bot}$. We want to
    show that $(\com{t}{\{u\}}, p+q)\in\Bot$.
    Two cases are possible:
    \begin{itemize}
      \item If $u$ is a value, then by $\redMu$-saturation of $\Bot$, it suffices to show that 
    $(\com{\mu \kappa. \com{t}{\{\kappa\}}}{u}, p+q)\in\Bot$.
      \item If $u$ is not a value then $\com{t}{\{u\}} \redi \com{\mu \kappa.\com{t}{\{\kappa\}}}{u}$
      and so by $\redi$-saturation of $\Bot$, it suffices to show that 
    $(\com{\mu \kappa. \com{t}{\{\kappa\}}}{u}, p+q)\in\Bot$.
      \end{itemize}
     In both cases it is a consequence
    of $(\mu \kappa.\com{t}{\{\kappa\}}, p)\in X^{\bot\bot\bot} = X^{\bot}$, which 
    is immediate because $(t,p)\in (\downarrow X)^{\bot}$ and by $\redMu$-saturation
    of $\Bot$. 
    
    \item Similarly we only have to prove that 
    $(X \otimes Y)^{\bot}\subseteq (X^{\bot\bot} \otimes Y^{\bot\bot})^{\bot}$. 
    Let $(t,q)\in ( X\otimes Y)^{\bot}$ and $((u,u'),p+p')\in  (X^{\bot\bot}\otimes Y^{\bot\bot})$. 
    By the same argument of $\redMu$ and $\redi$-saturation of $\Bot$, it suffices
    to prove that $(\mu \kappa.\com{t}{(\kappa,u')},q+p')\in X^{\bot}$. 
    To do so lets take $(u'',p'') \in X$ and show 
    $(\com{u''}{\mu\kappa.\com{t}{(\kappa,u')}}, q+p'+p'')\in \Bot$. Again, by
    $\redMu$ and $\redi$ saturation it suffices to show that 
    $(\mu\kappa'.\com{t}{(u'', \kappa')}, q+p'')\in Y^{\bot}$. 
    Again, take $(u''',p''')\in Y$. We have clearly 
    $(\com{t}{(u'',u''')}, q+p''+p''')\in\Bot$ since $(t,q)\in (X\otimes Y)^{\bot}$. 
    Hence our result. 
    \item This is immediate by Property \ref{prop:ortho_cap}.
    \item This is immediate by Property \ref{prop:ortho_cap}.
  \end{enumerate}
\end{proof}

\begin{lem}\label{lm:dbotopen}
	Let $A$ a formula, $t$ a term of the same polarity as $A$ and which has
	exactly one free variable $\kappa$, $q\in\M$
	and $X$ a subset of $\PosTerms\cap\mathbb{V}\times\M$ or of $\NegTerms\times\M$.
	The following properties are equivalent:
  \begin{enumerate}[i)]
    \item For each $(V,q)\in X$, $(c[V/\kappa], p+q)\in \Bot$
    \item For each $(V,q)\in X^\dbot\cap\mathbb{V}$, $(c[V/\kappa],p+q)\in \Bot$
  \end{enumerate}
\end{lem}
\begin{proof}
	\begin{itemize}
		\item $(ii)\Rightarrow (i)$ This is immediate since $X\subseteq X^\dbot\cap\mathbb{V}$.
		\item $(i)\Rightarrow (ii)$ Suppose $(i)$. That means for every $(V,q)\in X$,		
		by $\rightarrow_\mu$ saturation of $\Bot$, $(\com{\mu \kappa.c}{V},p+q)\in\Bot$. 
		Hence, $(\mu \kappa.c, p) \in X^{\bot} = X^{\dbot\bot}$, so 
		for every $(V,q)\in X^{\dbot}\cap \mathbb{V}$, $(\com{\mu\kappa.c}{V}, p+q)\in\Bot$
		and by $\rightarrow_\mu$ closure of $\Bot$ and because $V$ is a value, $(c[V/\kappa], p+q)\in\Bot$.
  \end{itemize}
\end{proof}

\begin{rmk}
	This property, which is true for every choice of saturated quantitative pole, will be 
	useful when we will extend our interpretation to stronger type systems
	(that can handle contraction). This lemma requires to work
	in a calculus where substitution and interaction can be \emph{exchanged} in the following
	sense:
	$$(c[t/\kappa], p+q)\in \Bot \Longleftrightarrow (\com{\mu \kappa.c}{t},p+q)\in\Bot$$
	  In particular, it is impossible to have this property in the framework
  of the usual Krivine's realizability: only head contexts are considered
  while we need general contexts. This justifies the 
  use of a completely symmetric calculus. 
\end{rmk}

\subsection{Adequacy} \label{subsec:adequacy} 

Before we can state and prove the soundness of our realizability
interpretation with respect to \class{MAL}$\omega$, we define what it means for a typing
rule to be \emph{adequate}. All the following notions are defined \emph{with respect to
some quantitative pole $\Bot$}.

By abuse, if $\Gamma = \kappa_1 : A_1,\dots, \kappa_n : A_n$
is a typing context and $\rho$ is a valuation, we will write
$\rho \Vdash \Gamma$ as a notation for $\rho \Vdash A_1,\dots, A_n$. 
We will also denote by $\Gamma[\rho]$ a pair $(\Gamma, \rho)$
where $\rho \Vdash \Gamma$.

\begin{defn}[Substitution]
A \emph{substitution} $\sigma$ is a partial application from the set
of term variables to the set $\mathbb{V} \times \M$,
whose domain $dom(\sigma)$ is finite. 
We will note $[\kappa_1 \leftarrow (V_1,p_1), \dots, \kappa_n \leftarrow (V_n,p_n)]$
to denote the substitution $\sigma$ where $dom(\sigma) = \{\kappa_1,\dots,\kappa_n\}$
and such that $\sigma(\kappa_i) = (V_i,p_i)$. 

\end{defn}
Suppose $\sigma = [\kappa_1 \leftarrow (V_1, p_1), \dots, \kappa_n \leftarrow (V_n, p_n)]$
is a substitution.
\begin{itemize}
  \item If $(c,q)$ is a bounded command, we note 
$$(c,q)[\sigma] = (c[V_1/\kappa_1, \dots, V_n/\kappa_n], q+\sum_i p_i)$$
  \item If $(u,q)$ is a bounded term then we note $$(u,q)[\sigma] = (u[V_1/\kappa_1,\dots, V_n/\kappa_n], q+\sum_i p_i)$$
  \end{itemize}

If $\sigma$ is a substitution, we denote by $\sigma[\kappa\leftarrow (V,p)]$ the substitution
obtained from $\sigma$ by rebinding $\kappa$ to $(V,p)$. If $\sigma_1$ is a subtitution
and $\sigma_2 = [\kappa_1 \leftarrow (V_1,p_1),\dots,\kappa_n \leftarrow (V_n,p_n)]$
is another substitution, we denote by 
$$\sigma_1,\sigma_2 = (\dots(\sigma_1[\kappa_1\leftarrow (V_1,p_1)])\dots )[\kappa_n \leftarrow (V_n,p_n)]$$

\begin{defn}[Adequate substitution]
Let $\Gamma= \kappa_1 : A_1 ,\dots,\kappa_n : A_n$ be a context 
and $\rho \Vdash \Gamma$.
We say that a substitution $\sigma$ is \emph{adequate} to $\Gamma[\rho]$
and we note $\sigma \Vdash \Gamma[\rho]$ iff 
\begin{itemize}
  \item $dom(\sigma) = \{ \kappa_1,\dots, \kappa_n\}$
  \item $\forall i \in \Val{1,n}, \sigma(\kappa_i) \in \Inte{A_i^{\bot}}_\rho$
  \end{itemize}
\end{defn}

In particular, if $\Gamma$ is a typing context and $\sigma$
a substitution adequate to $\Gamma$ then for every
negative (resp. positive) variable $\kappa$ appearing in $\Gamma$,
$\sigma(\kappa) \in \PosTerms\cap\mathbb{V}\times\M$ (resp. $\NegTerms\times\M$).
Indeed, positive (resp. negative) variables of $\Gamma$ are associated
to negative (resp. positive) formulas. 

\begin{defn}[Adequate judgment]
Suppose $\Gamma = \kappa_1 : A_1, \dots, \kappa_n : A_n$ is
a context and $p\in \M$. 
\begin{itemize}
  \item 
A judgment of the form $c:  (\vdash \Gamma)$ is said to be $p$-\emph{adequate} iff
for every total valuation $\rho$ and for every
adequate substitution $\sigma \Vdash \Gamma[\rho]$
we have  $(c,p)[\sigma] \in \Bot$. 
  \item 
  Similarly, a judgment of the form $\vdash t : B\delb \Gamma$ is said to be $p$-\emph{adequate} iff
for every total valuation $\rho$ and for every adequate substitution 
$\sigma \Vdash (\Gamma) [\rho]$ we have  $(t,p)[\sigma] \in \interp{B}_\rho$.
Moreover, if  $t \in \mathbb{V}$ then $(t,p)[\sigma] \in \overline{\Inte{B}_\rho}$. 
\end{itemize}
\end{defn}

We have now everything we need to define what it means
for our interpretation to be sound with respect to to a given typing rule.
A typing rule is given by a sequence of premises $J_i$ (which are
typing judgments), side-conditions (SC) on these judgments, and a conclusion $K$:
$$
\prooftree
 J_1\quad J_2\quad\dots  \quad J_n\quad\quad SC
 \justifies K
\thickness=0.08em
\shiftright 0em\using (rule) \endprooftree $$

\begin{defn}[Adequate rule]
Suppose $R$ is a typing rule, with 
$J_1, \dots, J_n$ being its premises judgments and $K$ its conclusion. 
Suppose $f : \M^{n} \rightarrow \M$ is a $n$-ary function on the quantitative monoid. 
We say that:

$R$ is $f$-\emph{adequate} iff for every $p_1,\dots, p_n\in \M$, (for all 
$i\in \Val{1,n}$, 
$J_i$ is $p_i$-adequate) implies that $K$ is $f(p_1,\dots, p_n)$-adequate.
\end{defn}

\begin{rmks}$ $ 
\begin{enumerate}
\item If a $0$-ary rule (like the axiom rule) then the notion of $f$-adequacy
makes sense only if $f$ is an element of the quantitative monoid $\M$. 
\item If a typing derivation $\pi$ is built using only adequate rules, then 
its conclusion is also $p$-adequate for some $p\in \M$. That $p$ is obtained
by composing the functions associated to each rule accordingly to the
derivation structure $\pi$. 
\end{enumerate}
\end{rmks}

We now prove an \textbf{adequacy theorem} that relates 
typing in \HOMAL and quantitative realizability. We suppose
having chosen a saturated quantitative pole 
$(\M, \Bot, \pBeta)$ and a
positive propositional domain 
$\Dpos$.
We first associate to each \HOMAL rule $R$ a function $\Mp{R} : \M\rightarrow \M$.
We will then show that each rule $R$ of \HOMAL is $\Mp{R}$-adequate. 
$$
\begin{array}{rclcrcl}
  \Mp{Ax_+}          & = & \zero & &                           \Mp{Ax_-} & = &\zero\\
  \Mp{\mu_-}           & = & x \mapsto x  & &     \Mp{\mu_+}           & = & x \mapsto x  \\
  \Mp{\otimes}      & = & (x,y) \mapsto x +y &  & \Mp{\parr} & = & x\mapsto x + \pBeta\\
  \Mp{\downarrow} & = & x\mapsto x & &                   \Mp{\uparrow} & = & x\mapsto x+\pBeta\\
  \Mp{\exists}       & = & x\mapsto x & &                   \Mp{\forall} & = & x\mapsto x\\
  \Mp{\congE}       & = & x\mapsto x  & &                  \Mp{W} &  =  & x\mapsto x \\
   \Mp{Cut}   & = & (x,y)\mapsto x+y & & & & \\
\end{array}
$$

\begin{thm}\label{thm:adequacy1}
Suppose that $(\M, \Bot, \pBeta)$ is a saturated quantitative pole. 
Every rule $R$ of \HOMAL is $\Mp{R}$-adequate. 
\end{thm}

\begin{proof}
To prove this statement, we just look at each of the \HOMAL rules and check that
they are adequate. 
\begin{itemize}
  \item $(Ax_*)$ The proof is the same in the positive and negative cases.
  Let $(t,p) \in \Inte{A^{\bot\bot}} = \Inte{A}$, then 
  $(\kappa,\zero)[\kappa\mapsto (t,p)] = (t,p) \in \Inte{A}_\rho \subseteq \overline{\Inte{A}_\rho}$.  Hence the rule is $\zero$-adequate. 
  
  \item $(Cut)$ Suppose $\vdash t_+ : P \delb \Gamma$ is $p$-adequate and
  $\vdash t_- : P^{\bot} \delb \Delta$ is $q$-adequate. Let $\rho \Vdash \Gamma,\Delta$
  and $\sigma \Vdash (\Gamma,\Delta)[\rho]$. We can split $\sigma = \sigma_1,\sigma_2$
  such that $\sigma_1 \Vdash \Gamma[\rho]$ and $\sigma_2\Vdash \Delta[\rho]$. 
  We want to show that $(\com{t_+}{t_-}, p+q)[\sigma]\in \Bot$.
  For any $\rho' \Vdash P,\Gamma,\Delta$ whose restriction is $\rho$
  (such a $\rho'$ exists), we have 
  $(t_+, p)[\sigma_1] \in \Inte{P}^{\bot\bot}_{\rho'}$ and $(t_-, q)[\sigma_2]\in \Inte{P^{\bot}}_{\rho'}
   = \Inte{P}_{\rho'}^{\bot}$. Hence $(t_+,p)[\sigma_1] \bot (t_-,q)[\sigma_2]$, and so
   $(\com{t_+}{t_-}, p+q)[\sigma]\in\Bot$. 
  
  \item $(\otimes) $ We suppose $\vdash V : A \delb \Gamma$ is 
  $p$-adequate and $\vdash V' : B \delb \Delta$ is $q$-adequate. 
  Let $\rho \Vdash \Gamma,\Delta, A\otimes B$ and $\sigma \Vdash (\Gamma,\Delta)[\rho]$. 
  Hence, $\sigma$ can be split into $\sigma_1 \Vdash \Gamma[\rho]$ and 
  $\sigma_2\Vdash \Delta[\rho]$.  Because $\rho \Vdash \Gamma,A$ and $\rho \Vdash \Delta,B$ 
  we know by hypothesis that $(V, p)[\sigma_1] \in \overline{\Inte{A}_\rho}$ and
  $(V', q)[\sigma_2] \in \overline{\Inte{B}_\rho}$. Hence $((V,V'), (q+q'))[\sigma_1,\sigma_2] \in 
  \overline{\Inte{A}_\rho\otimes\Inte{B}_\rho} = \overline{\Inte{A\otimes B}_\rho}$. Because $\sigma = 
  \sigma_1,\sigma_2$, we can conclude that the $(\otimes)$ rule
  is $\Mp{\otimes}$-adequate. 
  
  \item $(\parr) $ Suppose $c : (\vdash \kappa : A , \kappa' : B, \Gamma)$ is
  $p$-adequate for some $p$. Let $\rho \Vdash \Gamma, A\parr B$ and 
  $\sigma \Vdash \Gamma[\rho]$. We want to show that 
  $(\mu (\kappa,\kappa').c, p+\pBeta)[\sigma] \in \Inte{A\parr B}_\rho$. 
  Since $\Inte{A\parr B}_\rho = \Inte{A^{\bot}\otimes B^{\bot}}_\rho^{\bot}$ by
  Property \ref{prop:interp1}, we take $(V,q)\in \Inte{A^{\bot}}_\rho$
  and $(V',q') \in \Inte{B^{\bot}}_\rho$ and show that
  $(\com{\mu (\kappa,\kappa').c}{(V,V')}, p+\pBeta+q+q')[\sigma]\in\Bot$. 
  But it is easy to see that $\sigma,\kappa\mapsto (V,q),\kappa'\mapsto (V',q')\Vdash 
  (\Gamma, \kappa:A,\kappa':B)[\rho]$. 
  Hence, because the premise is $p$-adequate we obtain 
  $(c[V/\kappa, V'/\kappa'], p+q+q')[\sigma]\in\Bot$. By 
  $\redBeta$-saturation of $\Bot$, we finally obtain
  $(\com{\mu (\kappa,\kappa').c}{(V,V')}, p+\pBeta+q+q')[\sigma]\in\Bot$.

    \item $(\downarrow) $ We suppose $\vdash V : A \delb \Gamma$ is 
  $p$-adequate.
  Let $\rho \Vdash \Gamma, \downarrow A$ and $\sigma \Vdash \Gamma[\rho]$. 
  Because $\rho \Vdash \Gamma,A$ 
  we know by hypothesis that $(V, p)[\sigma] \in \overline{\Inte{A}_\rho}$. Hence $(\{V\}, q)[\sigma] \in 
  \overline{\downarrow \Inte{A}_\rho}  =   \overline{\Inte{\downarrow A}_\rho}$. 
  We conclude that the $(\downarrow)$ rule
  is $\Mp{\downarrow}$-adequate. 
  
   \item $(\uparrow)$  Suppose $c : (\vdash \kappa : A, \Gamma)$ is
  $p$-adequate for some $p\in \M$. Let $\rho \Vdash \Gamma, \uparrow A$ and 
  $\sigma \Vdash \Gamma[\rho]$. We want to show that 
  $(\mu \{\kappa\}.c, p+\pBeta)[\sigma] \in \Inte{\uparrow A}_\rho$. 
  Since $\Inte{\uparrow A}_\rho = \Inte{\downarrow A^{\bot}}_\rho^{\bot}$ by
  Property \ref{prop:interp1}, we take $(V,q)\in \Inte{A^{\bot}}_\rho$
  and show that
  $(\com{\mu \{\kappa\}.c}{\{V\}}, p+\pBeta+q)[\sigma]\in\Bot$. 
  But it is easy to see that $\sigma,\kappa\mapsto (V,q)\Vdash 
  (\Gamma, \kappa:A)[\rho]$. 
  Hence, because the premise is $p$-adequate we obtain 
  $(c[V/\kappa], p+q)[\sigma]\in\Bot$. By 
  $\redBeta$-saturation of $\Bot$, we finally obtain
  $(\com{\mu \{\kappa\}.c}{\{V\}}, p+\pBeta+q)[\sigma]\in\Bot$.

   \item $(\mu_+)$ 
    Suppose $c: (\vdash \alpha : P, \Gamma)$ is $p$-adequate
   for some $p\in \M$.  Let $\rho \Vdash P,\Gamma$ and $\sigma \Vdash 
   \Gamma[\rho]$. We want to show that $(\mu \alpha.c, p)[\sigma] \in 
   \interp{P}_\rho = \Inte{P}_\rho^{\bot\bot}$.
   So we take  $(u,q)\in \Inte{P}_\rho^{\bot}$
        and want to conclude that $(\com{\mu \alpha.c}{u}, p+q)[\sigma]\in \Bot$.         
        But $\Inte{P}^{\bot}_\rho =  \Inte{P^{\bot}}_\rho$
       by Property \ref{prop:interp1}.  Hence, $\sigma, \alpha\mapsto (u,q)\Vdash
       (\alpha:P,\Gamma)[\rho]$. Since the premise of the rule is $p$-adequate
       we conclude that $(c[u/\alpha], p+q)[\sigma]\in \Bot$. But
       $\com{\mu \alpha.c}{u} \redMu c[u/\alpha]$ because $\alpha$ and $u$
       are negative. Hence, by $\redMu$-saturation of $\Bot$ we obtain 
       $(\com{\mu \alpha.c}{u}, p+q)[\sigma]\in \Bot$. 

  \item $(\mu_-)$ 
    Suppose $c: (\vdash x : N, \Gamma)$ is $p$-adequate
   for some $p\in \M$.  Let $\rho \Vdash N,\Gamma$ and $\sigma \Vdash 
   \Gamma[\rho]$. We want to show that $(\mu x.c, p)[\sigma] \in 
   \Inte{N}_\rho = \Inte{N^{\bot}}_\rho^{\bot}$. Let $(u,q)\in 
   \Inte{N^{\bot}}_\rho$. It is sufficient to show that
    $(\com{\mu x.c}{u}, p+q)[\sigma]\in \Bot$.   
    But it is immediate that $\sigma,x\mapsto (u,q) \Vdash (x:N,\Gamma)[\rho]$.
    Hence, because the premise is $p$-adequate, we obtain 
    $(c[u/x],p+q)[\sigma] \in \Bot$. Since $u$ is a value (because
    $(u,q)\in \Inte{N^{\bot}}_\rho$), by $\redMu$ saturation 
    we obtain $(\com{\mu x.c}{u}, p+q)[\sigma]\in \Bot$. 
     
 \item $(W)$ Suppose $c : (\vdash \Gamma)$ is $p$-adequate. Let 
 $\rho \Vdash \Gamma, A$ and $\sigma \Vdash (\Gamma,\kappa: A)[\rho]$. 
 Then $\sigma = \sigma',\kappa\mapsto (u,q)$ with $\sigma' \Vdash \Gamma[\rho]$. 
 But $\rho \Vdash \Gamma$ so we conclude that 
 $(c,p)[\sigma'] \in \Bot$. By $\leq$-saturation of $\Bot$ we obtain immediately
 $(c,p)[\sigma]\in \Bot$.

 \item $(\forall^\tau)$ Suppose $\vdash V : A \delb \Gamma$ is $p$-adequate
 and $x^{\tau}$ does not appear free in $\Gamma$. We want to show that
 $\vdash V: \forall x^{\tau}.A \delb \Gamma$ is $p$-adequate. 
 Let $\rho \Vdash \forall x^{\tau}.A,
 \Gamma$ and $\sigma \Vdash \Gamma[\rho]$. By Remarks \ref{rmk:interp1} and because
 $V$ is a value, whatever
 the polarity of $A$ is, we have to show that 
 $(V,p)[\sigma] \in \bigcap_{v\in \Inte{\tau}} \overline{\Inte{A}_{\rho[x^{\tau}\leftarrow v]}}$. 
 So let $v\in \Inte{\tau}$ and we pose $\rho' =
 \rho[x^{\tau}\leftarrow v]$. We have $\rho' \Vdash (A,\Gamma)$.
 Moreover, $\sigma \Vdash \Gamma[\rho']$ because $\sigma\Vdash \Gamma[\rho]$
 and $x^{\tau}$ does not appear free in $\Gamma$. 
 Hence by hypothesis, $(V,p)[\sigma] \in \overline{\Inte{A}_{\rho[x^{\tau}\leftarrow v]}}$, which
 permits to conclude. 
 
 \item $(\exists^\tau)$ Let's first handle the case of values. 
  Suppose $\vdash V : A[T/x^{\tau}] \delb\Gamma$ 
 is $p$-adequate for some $T : \tau$. We want to show that 
 $\vdash V : \exists x^{\tau}.A \delb \Gamma$ is $p$-adequate. Let $\rho$ be a total valuation
  and $\sigma \Vdash \Gamma[\rho]$. We can suppose that $x^{\tau}$
 does not appear in $\Gamma$ (if it does, then we can rename it in $\exists x^{\tau}.A$). 
 Because of $p$-adequacy of the premise, we have $(V,p)[\sigma] \in \overline{\Inte{A[T/x^{\tau}]}_\rho}$. 
 But by Property \ref{prop:interp1}, we have $\Inte{A[T/x^{\tau}]}_\rho = 
 \Inte{A}_{\rho[x^{\tau}\leftarrow \Inte{T}_\rho]}$. 
 Hence $(V,p)[\sigma]\in \bigcup_{v\in \Inte{\tau}} \overline{\Inte{A}_{\rho[x^{\tau}\leftarrow v]}}
 \subseteq \overline{\Inte{\exists x^{\tau}.A}_\rho}$. 

 \item $(\exists^\tau)$ We now prove the case where $t$ is not a value (hence, the
 formula is positive). 
  Suppose $\vdash t : P[T/x^{\tau}] \delb\Gamma$ 
 is $p$-adequate for some $T : \tau$. We want to show that 
 $\vdash t : \exists x^{\tau}.P \delb \Gamma$ is $p$-adequate. Let $\rho$ be a total valuation
  and $\sigma \Vdash \Gamma[\rho]$. We can suppose that $x^{\tau}$
 does not appear in $\Gamma$ (if it does, then we can rename it in $\exists x^{\tau}.P$). 
 Because of $p$-adequacy of the premise, we have $(t,p)[\sigma] \in 
 \Inte{P[T/x^{\tau}]}_\rho^{\bot\bot}$. 
 But by Property \ref{prop:interp1}, we have $\Inte{P[T/x^{\tau}]}^{\bot\bot}_\rho = 
 \Inte{P}_{\rho[x^{\tau}\leftarrow \Inte{T}_\rho]}^{\bot\bot}$. 
 Hence $(t,p)[\sigma]\in \bigcup_{v\in \Inte{\tau}} \Inte{P}_{\rho[x^{\tau}\leftarrow v]}^{\bot\bot}
 \subseteq \Inte{\exists x^{\tau}.P}_\rho^{\bot\bot}$. 

\end{itemize}

\end{proof}

If $\pi$ is a typing derivation, we define $\Mp{\pi}$ as the element of $\M$
obtained by composing the $\Mp{R}$ of each rule appearing in $\pi$ in the obvious way. 
Hence, if $\pi$ is a typing derivation of \HOMAL, then Theorem \ref{thm:adequacy1}
says that its conclusion is 
$\Mp{\pi}$-adequate. 

\begin{rmk}
To prove the adequacy theorem, we crucially rely on the saturation properties
of $\Bot$. This is where we really need to have a saturated quantitative pole.
\end{rmk}

\subsection{Non quantitative Krivine's classical realizability}\label{subsec:nonquantreal}

In a particular case of our general quantitative classical realizability, it is possible
to validate a contraction rule (hence dropping the linearity constraint) and 
recover the non quantitative version of Krivine's classical realizability \cite{MunchReal},
which we call \emph{simple realizability}. 
It will be used in Section \ref{sec:forcing} to state the forcing decomposition of 
quantitative realizability. 

\begin{defn}
We note $\M_0$ the only 
quantitative monoid (without unit) whose underlying set is $\{ 0 \}$ equipped with the usual 
addition on natural numbers.
\end{defn}

Suppose we have a \textit{non quantitative pole} $\Bot \subseteq \ComSet$, that is a set of commands such that:
$$c'\in\Bot\text{ and }c\redh c'\text{ implies }c\in\Bot$$

Then we define a \textit{quantitative extension} of $\Bot$:
$$\Bot_0 =\setp{(c,0)}{c\in\Bot}$$

\begin{prop}
The structure $(\M_0, \Bot_0, 0)$ is a saturated quantitative pole.
\end{prop}
\begin{proof}
$\M_0$ is clearly a quantitative monoid. Moreover, if 
$c \redh c'$ and $(c', 0) \in \Bot_0$, we have $c' \in \Bot$
hence $c\in \Bot$. That implies $(c,0)\in \Bot_0$, which proves
both the $\redBeta$ and $\redMu$-saturation properties. Finally,
the $\leq$-saturation is immediate, since $\M_0$ is a singleton. 
\end{proof}

This quantitative pole induces an interpretation function $\Inte{.}_{\rho}$
and a realizability relation $\Vdash^{\rho}_{\Bot_0}$.
We define the simple realizability relation $\Vdash$ as:
$$t\Vdash^\rho T \Longleftrightarrow (t,0) \Vdash^{\rho}_{\Bot_0} T$$

If we add the following contraction rule to \class{MAL}$\omega$,
we obtain a formulation of  \class{PA}$\omega$ (higher-order Peano arithmetic): 

$$\prooftree
 c:(\vdash \kappa_1 : A, \kappa_2 : A, \Gamma) \justifies  c[\kappa/\kappa_1, \kappa/\kappa_2] : (\vdash \kappa: A, \Gamma)
\thickness=0.08em
\shiftright 0em\using (C) \endprooftree $$

Although it does not hold in general, when we use $\Bot_0$ and pose
$\Mp{C} = x\mapsto x$, the contraction rule is adequate:

\begin{prop}\label{prop:dummy_contr}
The rule $C$ is $\Mp{C}$-adequate.
\end{prop}
\begin{proof}
Suppose $c:(\vdash \kappa_1 : A, \kappa_2 :A, \Gamma)$ is $0$-adequate. 
Let $\rho \Vdash \Gamma, A$ and $\sigma \Vdash (\Gamma, \kappa : A)[\rho]$.
Then $\sigma = \sigma', \kappa \mapsto (u,q)$. $\M_0$ is a singleton so $q= 0$.
If we pose $\tau = \sigma',\kappa \mapsto (u,0), \kappa_2 \mapsto (u,0)$ then
 $\tau \Vdash (\Gamma, \kappa_1 : A, \kappa_2 : A)[\rho]$. Hence, 
 by adequacy of the premise, $(c,0)[\tau] \in \Bot_0$. Since $0 + 0 = 0$, we
 also have $(c[\kappa/\kappa_1,\kappa/\kappa_2], 0)[\sigma] =  
 (c,0)[\tau]$. Hence the conclusion is $0$-adequate. 
\end{proof}

Hence, as a corolloary of Theorem \ref{thm:adequacy1}
and Property \ref{prop:dummy_contr}, we recover an adequacy theorem 
for \HOPA. 

\begin{thm}\label{thm:adequacy2}
All rules $R$ of \HOPA are $\Mp{R}$-adequate. 
\end{thm}

\begin{rmk}
In the case of the non-quantitative realizability, the notions of 
$p$-adequate judgments and $f$-adequate rules can be simplified.
We will say that a judgment is \emph{adequate} if it is $0$-adequate, and
a rule is \emph{adequate} if it is $(x\mapsto x)$-adequate. 
\end{rmk}

The following remark show that this version of the contraction rule is not $x\mapsto x$-adequate in 
general. 

\begin{rmk}\label{rmk:contraction}
When considering the general quantitative framework, this version of the contraction rule 
$C$ is never $f$-adequate for $f = x\mapsto x$ 
as soon as the quantitative pole meets the following conditions:
\begin{itemize}
  \item The quantitative monoid has a unit $\unit$ (for example the integers monoid) 
  \item There is a command $c$ such that $c : (\vdash_{\text{\HOMAL}} x : X, y : X, \Gamma)$,
  a valuation $\rho \Vdash X,\Gamma$, 
  a substitution $\sigma \Vdash \Gamma [\rho]$ and such that there is some 
  $(u,q) \in \rho(X)$ such that $(c[u/x,u/y], p+q)[\sigma]\notin \Bot$. 
\end{itemize} 
\end{rmk}
\begin{proof}
Suppose the rule $C$ is $f$-adequate for $f = x\mapsto x$. 
By Theorem \ref{thm:adequacy1} we know that the judgment 
$c : (\vdash_{\text{\HOMAL}} x : X, y : X, \Gamma)$ is $p$-adequate for some $p$. 
Hence, we have $(c[u/x,u/y], p+2.q)[\sigma] \in \Bot$ by $p$-adequacy of the
typing judgment and because $\sigma \Vdash \Gamma[\rho]$.  
Moreover, since $C$ is adequate, its conclusion must be $f(p)$-adequate. 
So we have $(c[u/x,u/y], p + q)[\sigma] \in \Bot$, which is contradictory with
the assumptions. 
\end{proof}

\subsection{Quantitative reducibility candidates model}\label{subsec:quant_cand}

In this subsection, we build a particular class of quantitative realizability 
models. By applying Theorem \ref{thm:adequacy1} on these models, we can 
prove a linear time termination property of \HOMAL programs. Later, this
result will be extended to more sophisticated systems that also enjoy
bounded-time termination properties. The construction
relies on the definition of a quantitative extension of
the well-known \textit{reducibility candidates} (defined by orthogonality,
as in \cite{LengrandMiquelClassical2008,girard1987linear}), which we call
\textit{quantitative reducibility candidates}. \newline

In the rest of this subsection, the quantitative monoid and the pole are such that:
\begin{itemize}
  \item The quantitative monoid is any quantitative monoid with unit $\M = (M, +, \zero,\norm{.},\unit)$. 
  \item The quantitative pole is the structure $(\M, \Bot, \pBeta)$ generalizing the one described
  in Example \ref{ex:quantpole}:
  		\begin{itemize}
		  				\item $\Bot = \setp{(c,p)}{\Time^{\redBeta}(c) \text{ is defined and }\Time^{\redBeta}(c)\leq \norm{p}}$
  				\item $\pBeta = \unit$ 
  		\end{itemize}
\end{itemize}

We now use the fact that our syntax can be extended: we suppose that
our two instruction sets $\Kpos$ and $\Kneg$ contain respectively 
the constants written $\daiPos$ and $\daiNeg$. These two constants play 
the same role as free variables in the usual reducibility candidates argument.
The only relevant properties of these new constants are:
\begin{prop}\label{prop:daiprop}
If $V_+\in \PosTerms$, $t_- \in \NegTerms$ and $\dai_1,\dai_2\in \{ \daiPos, \daiNeg \}$, then
\begin{enumerate}
  \item $\com{V_+}{\daiNeg} \nredh$
  \item $\Time^{\redBeta}(\com{t_-}{\daiPos}) \leq \Time^{\redBeta}(\com{t_-}{(\dai_1,\dai_2)})$
  \item $\Time^{\redBeta}(\com{t_-}{\daiPos}) \leq \Time^{\redBeta}(\com{t_-}{\{\dai_1\}}$
\end{enumerate}
\end{prop}
\begin{proof}
All these properties are immediate. 
\end{proof}

\begin{defn}[Quantitative reducibility candidates]
The set of \emph{positive quantitative reducibility candidates},
denoted by $\CanPos$, is the set of elements $X \in \mathcal{P}(\PosTerms\cap\mathbb{V}\times\M)$
such that:
\begin{enumerate}
  \item $\BV{(X^{\bot\bot})} = X$
  \item $(\daiPos,\zero) \in X^{\bot\bot}$
  \item $X^{\bot\bot} \subseteq \{ (\daiNeg, \zero) \}^{\bot}$
\end{enumerate}
The set 
$\CanNeg = \setp{X^{\bot}}{X\in \CanPos}$ is the set of \emph{negative quantitative reducibility
candidates}. The set of \emph{quantitative reducibility candidates} $\Can$ is the set
$\CanNeg\cup \CanPos$. 
\end{defn}
The following lemmas are used to prove that the set $\CanPos$ can be used as a positive propositional domain. We have to check every closure condition of Definition \ref{def:prop_domain}
 
\begin{lem}
Whenever $X,Y \in \Can$, then $X\otimes Y \in \CanPos$.
\end{lem}
\begin{proof}
\begin{itemize}
  \item We want to show that $(\daiPos,\zero) \in (X\otimes Y)^{\bot\bot}$. 
   Let's take some
  $(t_-, p)\in (X\otimes Y)^\bot$. By Lemma \ref{lem:gen_tensor_shift}, we have
  also 
  $(t_-, p) \in (X^{\bot\bot}\otimes Y^{\bot\bot})^{\bot}$. 
    Depending of the polarity of $X$ and $Y$, we know that 
    $\dai_1 \in X^{\bot\bot}$ and $\dai_2 \in Y^{\bot\bot}$ for some $\dai_1,\dai_2 \in \{\daiPos,\daiNeg\}$. So  $(\com{t_-}{(\dai_1,\dai_2)},p)\in\Bot$.
    But by Property \ref{prop:daiprop},  
    \begin{eqnarray*}
    	\Time^{\redBeta}(\com{t_-}{\daiPos}) & \leq & \Time^{\redBeta}(\com{t_-}{(\dai_1,\dai_2)}) \\
    		& \leq & \norm{p}
    \end{eqnarray*}
  we can conclude that $(t_-, p)\bot (\daiPos, \zero)$. 
  So $(\daiPos,\zero)\in (X\otimes Y)^{\bot\bot}$. 
  
	\item We now need to show that $X\otimes Y\subseteq \{(\daiNeg, \zero)\}^\bot$. 
	We know that $X\otimes Y$ only contains values, since $X,Y \in \Can$. 
	But now, it is easy to see that if $(V_+,p)\in X\otimes Y$, then 
	immediately $\com{V_+}{\daiNeg}$ does not reduce for $\redh$ and so $(\com{V_+}{\daiNeg},p)\in\Bot$. 
\end{itemize}
\end{proof}

\begin{lem}
Whenever $X \in \Can$, then $(\downarrow X)^\dbot \in \CanPos$.
\end{lem}
\begin{proof}
\begin{itemize}
  \item We want to show that $(\daiPos,\zero) \in (\downarrow X)^\dbot$. Let's take some
  $(t_-, p)\in  (\downarrow X)^\bot$. By Lemma \ref{lem:gen_tensor_shift}, we know that
  $(t_-,p)\in (\downarrow X^{\bot\bot})^{\bot}$. But, depending of the polarity of $X$, 
  we know that $(\com{t_-}{\dai},p)\in\Bot$
  with $\dai\in \{\daiPos,\daiNeg\}$. In any case, because 
  \begin{eqnarray*}
  	\Time^{\redBeta}(\com{t_-}{\daiPos})  & \leq & 
  \Time^{\redBeta}(\com{t_-}{\{\dai\}}) \\
  & \leq & \norm{p}
  \end{eqnarray*}, we can conclude that $(t_-, p)\bot (\daiPos, \zero)$. 
  So $(\daiPos,\zero)\in (\downarrow X)^\dbot$. 
  
	\item We now need to show that $(\downarrow X)^\dbot \subseteq \{(\daiNeg, \zero)\}^\bot$.
	By orthogonality properties, it suffices to show that $\downarrow X \subseteq
	 \{(\daiNeg, \zero)\}^{\bot}$. 
	Since $X$ contains only values, so does $\downarrow X$.
	But it is immediate that for any $(V_+,p)\in \downarrow X$, 
		 $\com{V_+}{\daiNeg}$ does not reduce for $\redh$ and so $(\com{V_+}{\daiNeg},p)\in\Bot$. 
	\end{itemize}
\end{proof}

\begin{lem}
Suppose $\emptyset\neq D \subseteq \CanPos$, then $\bigcap_{X\in D} X \in \CanPos$.
\end{lem}
\begin{proof}
Suppose $\emptyset\neq D\subseteq \CanPos$. 
\begin{itemize}
  \item By hypothesis, for each $X\in D$ we have $(\daiPos,\zero) \in X^{\bot\bot}$.
    So it is immediate that $(\daiPos,\zero)\in\bigcap_{X\in D} X^{\bot\bot} = (\bigcap_{X\in D} X)^{\bot\bot}$ by Lemma \ref{lem:gen_tensor_shift}.
	\item Because each $X\in D$ is such that $X\subseteq X^{\bot\bot} \subseteq \{(\daiNeg,\zero)\}^{\bot}$, it is clear that $\bigcap_{X\in D} X \subseteq \{(\daiNeg, \zero)\}^{\bot}$ (since $D$
	is not empty and contains only non-empty sets).	
  \end{itemize}
\end{proof}

\begin{lem}
Suppose $\emptyset\neq D \subseteq \CanPos$, then $\bigcup_{X\in D} X \in \CanPos$.
\end{lem}
\begin{proof}
Suppose $\emptyset\neq D\subseteq \CanPos$. 
\begin{itemize}
  \item By hypothesis, for each $X\in D$ we have $(\daiPos,\zero) \in X^{\bot\bot}$.
  Since $D$ is not empty we have
  $(\daiPos, \zero) \in \bigcup_{X\in D} X^{\bot\bot}\subseteq
  (\bigcup_{X\in D}X^{\bot\bot})^{\bot\bot}$. But this is equal to
  $(\bigcup_{X\in D} X)^{\bot\bot}$ by Lemma \ref{lem:gen_tensor_shift}. 
	\item Because each $X\in D$ is such that $X\subseteq X^{\bot\bot} \subseteq \{(\daiNeg,\zero)\}^{\bot}$, it is clear that $\bigcup_{X\in D} X \subseteq \{(\daiNeg, \zero)\}^{\bot}$.
  \end{itemize}
\end{proof}

These four lemmas permit us to conclude that if 
we choose the set $\CanPos$ as the positive propositional domain, then for each formula $A$
and each valuation $\rho \Vdash A$,
$\Inte{A}_\rho$ is a quantitative reducibility candidate.\newline

\subtitle{Linear time termination}  As an example, 
we show how to use the adequacy theorem on the quantitative reducibility candidates
model to prove complexity properties of terms
typable in \HOMAL. To do this, we need to choose a concrete quantitative monoid with unit:
we take the natural numbers quantitative monoid defined in Example \ref{ex:int_monoid}.
We obtain the following theorem:

\begin{thm} 
If $c: (\vdash  \Gamma)$, then $c$ normalizes for $\redh$ 
using at most $|c|$ $\redBeta$-steps.
\end{thm}
\begin{proof}
  Suppose that $\pi$ is a proof of
   $c:(\vdash \kappa_1 : A_1,\dots, \kappa_n : A_n)$. 
  We set:
  \begin{itemize}
    \item The quantitative monoid of integers $\N$. 
    \item The quantitative pole $\Bot_\Time$.
  \end{itemize}
  Let $\rho$ be a total valuation such that 
  $\rho(x^{\sortpropPos}) = \{ (\daiPos, \zero) \} \in \CanPos$
  (such a valuation exists).
  By Theorem \ref{thm:adequacy1}, we know that 
  for all $(V_i,q_i)\in \Inte{A_i}_\rho$, we have
  $$(c[V_1/\kappa_1,\dots, V_n/\kappa_n], \Mp{\pi}+\sum_i q_i) \in \Bot$$
  
  By Lemma \ref{lm:dbotopen}, for every $(W_i,p_i)\in\Inte{A_i}_\rho^{\dbot}\cap \mathbb{V}$,
  $$(c[W_1/\kappa_1,\dots, X_n/\kappa_n], \Mp{\pi}+\sum_i p_i) \in \Bot$$
  
  Since $\CanPos$ is a positive propositional domain, 
  we know that $(\dai_i, \zero)\in \Inte{A_i}_\rho^{\bot\bot}$,
  where $\dai_i\in \{\daiPos,\daiNeg\}$ depending of the polarity
  of $A_i$. 
  It implies that 
  $$(c[\dai_1/\kappa_1,\dots, \dai_n/\kappa_n], \Mp{\pi}) \in \Bot$$
  Hence,
  \begin{eqnarray*}
    \Time^{\redBeta}(c) & = & \Time^{\redBeta}(c[\dai_1/\kappa_1,\dots, \dai_n/\kappa_n])\\
     & \leq & \norm{\Mp{\pi}}
   \end{eqnarray*}
   
   But, it is easy to see that $\Mp{\pi} \leq |c|$. Hence
   $\Time^{\redBeta}(c) \leq |c|$. 
\end{proof}

\section{Extending the model: Soft Affine Logic}\label{sec:other_types}

So far, we have only treated the multiplicative fragment of \class{PA}$\omega$.
In this section, we show how to extend the realizability interpretation to
more substantial fragments. We take the particular example of Soft Affine Logic
(augmented with higher-order quantifiers and arithmetical operations, and noted \HOSAL). 
To do this, we need to find a suitable quantitative monoid: this is done by turning
the \textit{soft resource monoid} defined in \cite{DalLago20112029} into a quantitative
monoid. We then extend the adequacy theorem and the construction of quantitative reducibility candidates
to this new system, finally we prove a polynomial bounded-time normalization property. The same 
methodology can be applied without any trouble to all
the systems handled in \cite{DalLago20112029}.

\subsection{Soft affine logic}

\newcommand{\redb}{{ \rightarrow_! }}

Soft affine logic \cite{baillot2004soft}, is
a simple extension of multiplicative linear logic by mean of 
weak exponentials. It ensures a polystep normalization property
of programs typable in this system. \newline

We now suppose that the set $\Kpos$ of positive instructions contains
a term $!V$ for each value $V$. We also suppose that $\Kneg$ contains
a term $\mu !(\kappa).c$ for each command $c$ and variable $\kappa$. 
We also suppose that the reduction relation $\redu$ contains the following
binary relation $\redb$:
$$\com{\mu !(\kappa).c}{!V} \,\redb\, c[V/\kappa]$$

Suppose $\{x_1,\dots, x_k\}$ is a set of positive variables,
$t$ is a term and $\kappa$ is a fresh variable of the same polarity as $t$.
Then we define respectively a command $!_{\{x_1,\dots, x_k\}}^{\kappa} t$ and a term
$!_{\{x_1,\dots, x_k\}} t$ inductively as follows:
\begin{eqnarray*}
!_\emptyset^{\kappa} t & = & \com{!t}{\kappa}\\
!_{S\cup\{ x\} }^{\kappa} t & = & \com{\mu !(x). (!_S^{\kappa} t)}{x}\\
!_S t & = & \mu \kappa.(!_S^{\kappa} t)
\end{eqnarray*}
For example, with the variables $x_1,\dots, x_k$, we have
$$!_{\{x_1,\dots, x_k\}} t = \mu \kappa.\com{\mu !(x_1).  \com{\mu !(x_2).  \com{\dots .\com{!t}{\kappa}}{\dots} \dots }{x_2}}{x_1}$$

\begin{prop}\label{prop:bang_commute}
Suppose $x_1,\dots, x_k$ are positive variables
and $t$ is a term whose free variables are included in $\{ x_1,\dots, x_k\}$. 
Suppose moreover that $u_1, \dots , u_k$ are positive values
and $V$ is a closed value of the opposite polarity as $t$'s polarity. 
Then we have 
$$\com{!_{\{x_1,\dots, x_k\}} t}{V}[!u_1/x_1, \dots, !u_k/x_k] \redMu\redb^{k} \com{!t}{V}[u_1/x_1,\dots, u_k/x_k]$$
\end{prop}

\begin{rmk}
The construction $!_{\{x_1,\dots, x_k\}}$ is here to mimic
the functorial $!$ box of \class{SAL} proofnets. The last property
is then the counterpart of the reduction resulting
of the interaction between several functorial boxes. 
\end{rmk}

 We add two new formula constructors
$?A$ and $!A$. 
$$ A,B,T, U  ::=  \dots \delb !A\delb ?A$$

The formulas $!A$ and $?A$ are respectively positive and negative formulas,
in accordance with the following two new constructor typing rules:

\begin{center}
\begin{tabular}{cc}
\prooftree
    A : \sortprop \justifies !A : \sortpropPos
\thickness=0.08em
\shiftright 0em\endprooftree  &
\prooftree
   A : \sortprop \justifies ?A : \sortpropNeg
\thickness=0.08em
\shiftright 0em\using  \endprooftree\\
\end{tabular}
\end{center}

Typing rules of \HOSAL are then obtained by extending the rules 
presented in Subsection \ref{sec:HOMAL} with the following ones:
\begin{center}
\begin{tabular}{cc}
\prooftree
   \vdash V : A \delb x_1:N_1,\dots, x_k:N_k \justifies \vdash !_{\{x_1,\dots, x_k\}} V : !A \delb x_1 : ?N_1,\dots, x_k:?N_k
\thickness=0.08em
\shiftright 0em\using (!_k) \endprooftree  &
\prooftree
   c: (\vdash \kappa_1 : A, \dots, \kappa_n : A, \Gamma) \justifies \vdash \mu !(\kappa).c[\kappa/\kappa_1,\dots, \kappa/\kappa_n]: ?A \delb\Gamma
\thickness=0.08em
\shiftright 0em\using (M_n) \endprooftree\\
\end{tabular}
\end{center}

\begin{rmks}
$ $ 
\begin{enumerate}
  \item 
The choice of having only negative formulas in the context of the $!$ rule
is not restrictive, since we can always use the $\uparrow$ rule
to obtain such a context. But doing so allows not to care about
the polarity of variables. 

  \item
The multiplex rule $(M_n)$ and the promotion $(!_n)$ rules are
 in fact typing schemes,
that is  one rule for each integer $n\in\N$. In the case of the
promotion rule, this in fact accounts for the fact that 
functorial promotion is a cluster rule consisting of $n$ derelictions and
one usual promotion.
  \item
Notice that the multipliex rule, in the $n=1$ case is exactly
what we usually call the \emph{dereliction rule}. One could think
that it is possible to decompose the multiplex rule in two
more elementary rules: the dereliction rule
and the usual contraction. However, this would
lead to typable programs that can calculate functions
that are not computable in polynomial time. 
\end{enumerate}
\end{rmks}

If $\pi$ is a typing derivation in \HOSAL, then we define its \textit{depth}
$\delta(\pi)$ as the maximum number of nested $(!)$ rules appearing in it. 
In the rest of this paper,
we will use the symbol $\vdashSAL$ instead of $\vdash$ when we talk 
about typability in this new type system. 


\subsection{Soft monoid}

In order to obtain a model where those rules are adequate,
we need a richer structure than the quantitative monoid. 
This structure is given by the notion of \emph{soft exponential}. 
\def\pContr{{p_{contr}}}

\begin{defn}
Let $\M = (M, +, 0, \leq, \norm{.})$ be a quantitative monoid. Then a \emph{soft exponential}
on $\M$ is given by a family $(r_n)_{n\in\N}$ of elements  of $\M$ and an operation
$! : \M \rightarrow \M$ that satisfy the following properties:
\begin{itemize}
  \item For all $p,q\in \M$, we have $!p+!q \leq !(p+q)$.
  \item For all $p\in \M$ and $n\in \N$, we have $n.p \leq !p + r_n$.
\end{itemize}
\end{defn}

 We now give a concrete example of a quantitative monoid with unit and
 soft exponential. This monoid is obtained from the \textit{soft resource monoid}
described in \cite{DalLago20112029}. 

\begin{defn}
The \textit{soft monoid} is the structure $\M_s = (M_s, +_s, \zero_s, \unit_s, \leq_s, \norm{.}_s)$ where
\begin{itemize}
  \item $M_s$ is the set of pairs $(n, f)$ where $n\in \N$ and
  $f \in \N[X]$ is a polynomial with integer coefficients. 
  \item $(n,f)+_s (m,g) = (max(n,m), f+g)$ where $max(n,m)$ is the maximum of $n$ and $m$.
  \item $\zero_s = (0, x\mapsto 0)$ and $\unit_s = (0,  1)$.
  \item $(n,f) \leq_s (m,g)$ iff $n \leq m$ and $\forall x \geq m, f(x) \leq g(x)$ and
  $(g-f)(x)\leq (g-f)(y)$ for $m\leq x \leq y$ .
  \item $\norm{(n,f)}_s = f(n)$.
\end{itemize}
\end{defn}

\begin{prop}
$\M_s$ is a quantitative monoid with unit.
\end{prop}
\begin{proof}
\begin{itemize}
\item
It is clear that $(\M_s, +_s, \zero_s, \leq_s)$ is a preordered
commutative monoid. 
\item Let $(n,f)$ and $(m,g)$ be two elements of $M_s$. 
We have 
\begin{eqnarray*}
  \norm{(n,f)}+\norm{(m,g)} & = & f(n) + g(m)\\
  & \leq & f(max(n,m)) + g(max(n,m))\\
  & = & (f+g)(max(n,m))\\
  & = & \norm{(max(n,m),f+g)}\\
  & = & \norm{(n,f)+_s(m,g)}
\end{eqnarray*}

\item Suppose $(n,f)\leq_s (m,g)$. It means that $n\leq m$ and
$\forall x\in \N\text{ such that } x\geq m,  f(x)\leq g(x)$.  Hence, we have
$$\norm{(n,f)} = f(n) \leq f(m) \leq g(m) = \norm{(m,g)}$$

\item Finally $\unit$ is a unit, since $\norm{\unit}_s = 1$. 

\end{itemize}
\end{proof}

We moreover define the operation $! : M_s \longrightarrow M_s$ as
$!(n,f) = (n, f^+)$ where $f^+(X) = (X+1)f(X)$. This operation 
enjoys various properties.

\begin{prop}\label{prop:soft_monoid}
The pair $(!, \{(n,0)\}_{n\in\N})$ is 
a soft exponential. 
\end{prop}
\begin{proof}
Here, we pose $p = (n,f)$ and $q= (m,g)$. 
  \begin{enumerate}[(i)]
    \item We have 
    \begin{eqnarray*}
    	!(p+_s q)  & = &  (max(n,m), (f+g)^{+}) \\
				   & = & (max(n,m), (X+1)(f+g)) \\
				   & = & (max(n,m), (X+1)f + (X+1)g)\\	
				   & = & (max(n,m), f^{+}+g^{+})\\
				   & = & !p+_s !q
	\end{eqnarray*}
	\item We have 
    \begin{eqnarray*}
    	k.p  		& = &  (n, k.f) \\
					& \leq_s & (max(n,k), (X+1)f)\\
				   & = & (max(n,k), f^{+})\\
				   & = & !p+_s (k,0)
	\end{eqnarray*}
	\item Immediate.
	
\end{enumerate}
\end{proof}
      
Properties $(i)$ and $(ii)$ are crucial to obtain respectively 
monoidality of $!$ and the multiplexing rule (hence to prove adequacy).  \newline

\subsection{Interpretation of \HOSAL}

We now extend the realizability interpretation
defined on the multiplicative fragment to the exponentials.
We suppose that:
\begin{itemize}
  \item $\redu$ is an evaluation relation that contains $\redb$.
  \item $\M$ is a quantitative monoid with a soft exponential
  $(!,(r_n)_{n\in\N})$. 
  \item $(\M,\Bot,\pBeta)$ is a quantitative pole which is moreover
  $\redb$-saturated:
  $$\text{For every } (c,p)\in \Bot, \text{ if } c'\redb c\text{ then } (c',p+\pBeta)\in\Bot$$
\end{itemize}
  
We  introduce a new unary operation $!$ on sets of bounded terms:
$$!X = \setp{(!V,!p)}{(V,p)\in \BV{(X^{\bot\bot})}\times \M}$$

and we extend the interpretation of formulas as follows: 
\begin{eqnarray*}
\Inte{!A}_\rho & = &  !\Inte{A}_\rho\\
\Inte{?A}_\rho & = & (!\Inte{A^{\bot}}_\rho)^{\bot}
\end{eqnarray*}

As in Subsection \ref{subsec:adequacy}, to state the adequacy theorem
we need to associate a function to the two new rules:
\begin{eqnarray*}
  \Mp{!_n} & = & x\mapsto !x + n.\pBeta\\
  \Mp{M_n} & = & x\mapsto x + \pBeta+r_n
 \end{eqnarray*}

 \begin{thm}
 All rules $R$ of \HOSAL are $\Mp{R}$-adequate.
 \end{thm}

\begin{proof}
\begin{itemize}
  \item For all the multiplicative part, the proof is the same as 
  the one of the \HOMAL adequacy theorem. 
  \item $(!)$: Suppose $\vdash V : A \delb \Gamma$ is $p$-adequate. 
  Let $\rho$ be a total valuation and  $\sigma \Vdash ?\Gamma[\rho]$.
  We know that for each $x_i : N_i\in \Gamma$, we have 
  $\sigma(x_i) = (!V_i, !q_i)$ where $(V_i,q_i)\in \Inte{N_i^{\bot}}_\rho$. 
  If we pose $\sigma' = [x_1\leftarrow (V_1,q_1),\dots ,x_k \leftarrow (V_k, q_k)]$, 
  we have clearly $\sigma' \Vdash \Gamma[\rho]$. By hypothesis $(V,p)[\sigma'] \in \overline{\Inte{A}_\rho}$. 
  Hence, $(V,p)[\sigma'] \in \Inte{A}_\rho^{\bot\bot} \cap \mathbb{V}$ and finally
  $(!V[V_1/x_1,\dots, V_k/x_k], !(p+q_1+\dots+q_k)) \in \Inte{!A}_\rho$. 
  Because $!$ is a soft exponential on $\M$, and by $\leq$-saturation, we obtain
  $(!V[V_1/x_1,\dots, V_k/x_k], !p + !q_1 +\dots + !q_k)\in \overline{\Inte{!A}_\rho}$. 
  Finally, we know by Property \ref{prop:bang_commute}, $\redBeta$ and $\redMu$-saturation
  that 
  $$(!_{\{x_1,\dots, x_k\}} V[!V_1/x_1,\dots, !V_k/x_k], !p + k.\pBeta + !q_1 +\dots + !q_k)\in \overline{\Inte{!A}_\rho}$$ 
  which can be rewritten as
  $$(!_{\{x_1,\dots, x_k\}} V, !p + k.\pBeta)[\sigma] \in \overline{\Inte{!A}_\rho}$$ 
  Hence $\vdash !V : !A \delb \Gamma$ is $!p$-adequate and so the $!_k$ rule
  is $(x\mapsto !x + k.\pBeta)$-adequate. 
  
  \item $(Mplex_n)$: Suppose $c: (\vdash \kappa_1: A, \dots, \kappa_n : A, \Gamma)$
  is $p$-adequate, let $\rho$ be a total valuation and $\sigma \Vdash \Gamma$. 
  We want to show that $$(\mu !(\kappa).c[\kappa/\kappa_1,\dots, \kappa/\kappa_n], p+\pBeta+r_n)[\sigma]
  \in \Inte{?A}_\rho = \Inte{!A^{\bot}}_\rho^{\bot}$$

 First, notice that if 
	 $(V,q)\in \Inte{A}_\rho$, 
    if we pose $\sigma' = \sigma[\kappa_1\leftarrow V,\dots \kappa_n\leftarrow V]$, we clearly
    have $\sigma' \Vdash (\kappa_1:A,\dots, \kappa_n:A, \Gamma)[\rho]$. 
    By $p$-adequacy of the hypothesis, we obtain 
    $(c[V/\kappa_1,\dots, V/\kappa_n], p+n.q)[\sigma]\in \Bot$.

     Now, let $(V,q) \in  \Inte{A}^{\bot\bot}_\rho \cap \mathbb{V}$. 
    We want to show that $$\mu !(\kappa).c[\kappa/\kappa_1,\dots, \kappa/\kappa_n], 
    p+\pBeta)[\sigma] \bot (!V, !q)$$
     which will prove that the conclusion is $(p+\pBeta)$-
    adequate. But we know that 
    $\com{\mu !(\kappa).c[\kappa/\kappa_1,\dots, \kappa/\kappa_n]}{!V}
    \redh c[V/\kappa_1,\dots, V/\kappa_n]$. 
    But by the previous point, combined $n$ times with Lemma \ref{lm:dbotopen}, 
    we obtain 
    $$\forall (V',q')\in X^{\dbot}, (c[V'/\kappa_1,\dots, V'/\kappa_n], p+n.q')[\sigma]\in \Bot$$
    Hence $(c[V/\kappa_1,\dots, V/\kappa_n], p+n.q)\in \Bot$. 
    By $\leq$-saturation of $\Bot$ and because $!$ is a soft exponential, we have
    $(c[V/\kappa_1,\dots, V/\kappa_n], p+!q+r_n)\in \Bot$. 
    By $\redBeta$-saturation we finally obtain 
   $$(\mu !(\kappa).c[\kappa/\kappa_1,\dots, \kappa/\kappa_n], p+\pBeta+r_n)[\sigma]
  \in \Inte{?A}_\rho $$
\end{itemize}
\end{proof}

\begin{rmk}
In this proof, the lemma \ref{lm:dbotopen} is crucial to show
adequacy of the multiplex rule. The situation would be similar
for any system containing modality-rules that change the whole
context. 
\end{rmk}

\subsection{Polynomial bounded time termination}

We now prove the polynomial bounded time termination of \HOSAL, by
extending the technique of quantitative reducibility candidates. \newline

To obtain a bounded normalization theorem, we 
need to check that the construction of the quantitative
reducibility candidates is still valid. The definitions remain
the same as those of Subsection \ref{subsec:quant_cand}, except
for the definition of $\Bot$:
$$\Bot = \setp{(c,p)}{\Time^{{(\redBeta\cup\redb)}}(c)\text{ is defined and is
bounded by }\norm{p}}$$

 The only
new property we need is the following one:

\begin{lem}
If $X \in \Can$ then $!X \in \CanPos$. 
\end{lem}
\begin{proof}
Without any loss of generality, let's suppose $X$ is positive. 
\begin{itemize}
  \item  Let 
  $(t,p)\in (!X)^{\bot}$. We have to show that $(\daiPos,\zero)\bot (t,p)$. 
  But we know that $(\daiPos, \zero)\in X^{\dbot}$. 
  Hence, $(!\daiPos, !\zero)\in !X$. By Property \ref{prop:soft_monoid},
  $!\zero = \zero$. We then have $(\com{t}{!\daiPos}, p+\zero)\in \Bot$. 
  But we know that 
  \begin{eqnarray*}
  		\Time^{(\redBeta\cup\redb)}(\com{\daiPos}{t}) & \leq &  \Time^{{(\redBeta\cup\redb)}}(\com{!\daiPos}{t})\\
		 & \leq  & \norm{p}
  \end{eqnarray*}
  Hence $(\daiPos,\zero)\in (!X)^{\dbot}$. 
  \item Let $(V,p)\in !X$, because $V$ is a positive value,
  it is immediate that $\com{V}{\daiNeg}$ does not reduce for $\redh\cup\redb$. 
  Hence $!X \subseteq \{(\daiNeg,\zero)\}^{\bot}$ and so
  $(!X)^{\bot\bot} \subseteq \{(\daiNeg,\zero)\}^{\bot}$. 
 \end{itemize}
\end{proof}

 By instantiating the monoid with
the soft monoid, we can derive a polystep normalization property of
terms typable in \HOSAL which extends the linear time normalization property of
Subsection \ref{subsec:quant_cand}. 

\begin{thm} 
There exists a family $(P_k)_{k\in\N}$ of polynomials on $\N$
such that if $\pi$ is a proof of $\vdashSAL t : A \delb $, then $t$ normalizes 
in at most $P_{\delta(\pi)}(|t|)$ reduction steps.  
\end{thm}
\begin{proof}
The proof consists essentially to remark that in the definition of
$\Mp{\pi}$, the only rule that makes the degree of
the polynomial of $\Mp{\pi}$ rise is the $(!)$ rule. The other rules
cause only the linear part of $\Mp{\pi}$ to grow. 
\end{proof}

%
%
%
%
%
\newcommand{\CInt}[1]{\mathcal{C}[#1]}
\newcommand{\FInt}[1]{\CInt{#1}^\bot}

\newcommand{\dbar}[1]{{\forcingOrtho{\forcingOrtho{#1}}}}
\def\BotTrans{{\FTrans{\Bot}}}
\def\botTrans{{\FTrans{\bot}}}
\def\dbotTrans{{\botTrans\botTrans}}
\def\VdashTrans{{\,\FTrans{\Vdash}\,}}
\newcommand{\InteTrans}[2]{{\FTrans{\Inte{#1}_{#2}}}}
\newcommand{\InteQuant}[2]{{{\Inte{#1}^{\circ}_{#2}}}}
\newcommand{\interpTrans}[2]{{\FTrans{\interp{#1}_{#2}}}}
\def\redTrans{ {\rightarrow_\bullet} }
\def\nredTrans{ {\nrightarrow_\bullet} }
\def\EncNat{{\FTrans{\N}}}
\def\rhoTrans{{\FTrans{\rho}}}
\def\BotQuant{{\Bot^\circ}}
\def\botQuant{{\bot^\circ}}
\def\dbotQuant{{\botQuant\botQuant}}
\def\rhoQuant{{\rho^\circ}}
\def\VdashQuant{{\,\Vdash^\circ\,}}

\section{A forcing decomposition} \label{sec:forcing}

In this section, we exhibit a connection between our quantitative extension 
of classical realizability and certain forcing interpretations. 
We precisely show that by composing non-quantitative classical
realizability with 
a notion of forcing for \class{MAL}, we obtain an instance of quantitative
realizability. We finally show that quantitative reducibility candidates
are a special case of this construction.
We proceed with the following methodology:
\begin{enumerate}
  \item We define the notion of \emph{linear forcing structure}, a variation of
  the notion of forcing structure already defined in \cite{krivine2010realizability}.
  \item We introduce a \emph{forcing translation}, that is the formalization
  of a class of forcing model of \MAL \emph{inside} \HOMAL.
  The result
  is a relation $p\FVdash A$, parametrized by a choice
  of linear forcing structure.   
  \item We describe a new machine: the \emph{countdown machine}.
  It is based on the same term syntax as \LFOC, but with a different
  notion of command and different reduction rules. 
  This machine induces a new class of non-quantitative realizability relations
  for \HOMAL, parametrized by a set $\BotTrans$ and denoted $t \VdashTrans A$.
  \item We show that for a \emph{particular} instance of linear forcing
  structure and for \emph{every} choice of $\BotTrans$, there exists 
  a \emph{quantitative pole} (in the sense of Section \ref{sec:QuantKrivine}) $\BotQuant$
  such that the associated realizability relation $\VdashQuant$ satisfies for every
  \MAL formula $A$:
  $$(t,p) \VdashQuant A \Longleftrightarrow t \VdashTrans (p \FVdash A)$$
  That means composing $\FVdash$ and $\VdashTrans$ yields a quantitative
  model of \MAL 
  \item Finally, we show that quantitative reducibility candidates of Subsection 
  \ref{subsec:quant_cand}, restricted to \MAL, can be seen as the result
  of such a composition.
\end{enumerate}

This methodology could be used to study forcing translations of \class{MAL}$\omega$, but we believe
it is simpler to explain it with \class{MAL}. In the last subsection we explain how it could be extended to
the whole system \class{MAL}$\omega$. 

\subsection{Preliminaries}

We define some concepts and notations required to define the forcing
translation and establish the associated results. \newline

\newcommand{\toHO}[1]{{(#1)}^{\omega}}

\subtitle{Multipicative Affine Logic}  \MAL (Multiplicative Affine Logic) 
is the affine, second-order fragment of \HOMAL. The following
grammar defines \MAL formulas:
\begin{eqnarray*}
   A, B & ::=  & P \delb N\\
   P & ::= & X \delb A\otimes B\delb \downarrow A\\
   N & ::= & X^{\bot} \delb A\parr B\delb \uparrow A
\end{eqnarray*}

We suppose that to each \MAL variable $X$ we associate a 
\HOMAL variable $X^{\sortpropPos}$. 
Then any \MAL formula $A$ can be seen as a \HOMAL constructor 
$A^{\omega}$ of kind $\sortprop^{\circ}$ (with  $\circ \in \{ +, - \}$), defined as follows:
\begin{eqnarray*}
  \toHO{X} & = & X^{\sortpropPos}\\
  \toHO{X^{\bot}} & = & (X^{\sortpropPos})^{\bot} \\
  \toHO{A \otimes B} & = & \toHO{A} \otimes \toHO{B}\\
  \toHO{A\parr B} & = & \toHO{A} \parr \toHO{B}\\
  \toHO{\downarrow A} &  = & \downarrow \toHO{A}\\
  \toHO{\uparrow A} & = & \uparrow \toHO{A}
\end{eqnarray*}

We will abusively use the same notation $A$ for both the \MAL formula $A$
and its associated \HOMAL constructor $\toHO{A}$.
We will also use the notation $A \multimap B = A^\bot \parr B$, that is the call-by-name linear implication. \newline

\subtitle{Realizability relations} In the rest of this section, we will manipulate several
realizability relations. We will in particular consider a quantitative realizability 
relation (that will be denoted by $\VdashQuant$ in further subsections)
and a simple realizability relation (that will be denoted by $\VdashTrans$)
in the sense of the $\Vdash^{0}$ relation of Subsection \ref{subsec:nonquantreal}.
In both cases, it will be implicit that the interpretation of constructors remain
the same as the one defined in Section \ref{sec:QuantKrivine}.\newline

\subtitle{Equational implication} To define the formula translation, we follow \cite{Miquel2011Forcing}
and add a new type constructor
to \HOMAL. If $A$ is a formula of kind $\sortpropPos$ (resp. $\sortpropNeg$), 
and if $T$ and $U$ are two type constructors of the same
kind, then we have a new type constructor of kind $\sortpropPos$
(resp. $\sortpropNeg$) denoted
$\langle T = U \rangle A$.  Its informal meaning is ``$T \congE U$ implies $A$''.
Even if we can define the forcing translation without it, it will simplify
the proof of the connection lemma. 
We will not add the corresponding typing rules, because we will deal directly 
with (non quantitative) realizability. For any valuation $\rho$, the realizability 
interpretation of section \ref{sec:QuantKrivine} is extended to the new formula 
$\langle T = U \rangle A$ as follows:
$$|\langle T = U \rangle A |_\rho = \left\{
 \begin{array}{ll}
        |A|_\rho & \mbox{if } T \congE U\\
        \emptyset^\bot & \mbox{ else}
    \end{array}
\right.$$

\subsection{Linear forcing structures}\label{subsec:linear_structures}

To compose realizability and forcing, we formalize a forcing
interpretation \textit{inside} \class{MAL}$\omega$. We mostly follow Krivine's formulation
of forcing \cite{krivine2010realizability, Miquel2011Forcing}. We begin by 
giving a \textit{linear} version of Krivine's forcing structure. 

\begin{defn}[Linear forcing structure]
A \textit{linear forcing structure} is given by the following
components:
\begin{itemize} 
  \item $\kappa$, the kind of \textit{conditions}.
  \item $\CInt{.} : \kappa \rightarrow \sortprop^{\circ}$,
  with $\circ \in \{ +, - \}$ is a positive or negative predicate.
  \item $\zero : \kappa$  is a distinguished condition.
  \item $+ : \kappa \rightarrow \kappa \rightarrow \kappa$ is 
a binary operation on conditions, such that for every $p,q,r : \kappa$,
the following conversions hold in \HOMAL:
\begin{eqnarray*}
  p+(q+r) & \congE & (p+q)+r\\
  p+q     & \congE & q+p\\
  \zero+p & \congE & p
\end{eqnarray*}

\end{itemize}
\end{defn}

\begin{exa}
A simple example of linear forcing structure is the integer forcing structure,
defined as follows:
\begin{itemize}
	\item The kind of conditions is the kind $\iota$ of integers. 
	\item The predicate is $\lambda x. \top : \iota \rightarrow \sortpropNeg$. 
	\item $+$ is the usual addition on integers defined using $rec_\iota$. 
	\item $\zero$ is the constructor $\foZero : \iota$. 
\end{itemize}
Then it is clear that for every integers $p,q,r : \iota$, the requested 
conversions hold in \HOMAL. 
\end{exa}

\begin{rmk}
What we note $+$ and $\zero$ is written ${\bf .}$ and $\unit$ in \cite{krivine2010realizability, Miquel2011Forcing}.
We choose an additive notation instead of a multiplicative
one, because we think it better fits the quantitative intuition of multiplicative
linear logic. It has also the advantage of being closer to the symbols used in the
definition of
quantitative monoids. 
\end{rmk}

%
The linear forcing structure is a sharp simplification of
the notion of forcing structure as defined in \cite{krivine2010realizability},
in particular because we ask $\kappa$ to be a monoid with respect to $\congE$ (that is,
at the computational level instead of the provability level).
This is however sufficient for our purpose.

\begin{rmk}
Informally,
$\FInt{p+q}$ represents a notion of orthogonality 
between $p$ and $q$, and plays the same role
in forcing as the pole $\Bot$ in realizability.
Observe that a linear forcing structure (modulo $\congE$) is a multiplicative
phase space \cite{girard1987linear}, by choosing 
$\FInt{.}$ as the pole. 
\end{rmk}

\newcommand{\forcingOrtho}[1]{{\overline{#1}}}

\newcommand{\preforc}[1]{{#1}^{*}}

\subsection{Formula translation} 

We assume having fixed a linear forcing structure on the kind $\kappa$.
We now formalize inside \HOMAL a forcing interpretation of $\MAL$. 
Following \cite{Miquel2011Forcing} methodology, we associate to each $\MAL$ formula 
$A$ a \HOMAL  formula $p\FVdash A$ (which is read "$p$ \textit{forces} $A$"). Because
in the rest of this paper all the quantifications are made on $\kappa$, we omit to
indicate the kinds on the quantifiers and on the variables of kind $\kappa$. 

\begin{defn}
Let $Z : \kappa \rightarrow \sortprop^{\circ}$ (for $\circ \in \{+,-\}$).
Then the \emph{forcing orthogonal} of $Z$ is defined as a \HOMAL
constructor of kind $\kappa \rightarrow \sortpropNeg$:
$$\forcingOrtho{Z}\quad \defequ \quad\lambda r.\forall r'.Z(r')\multimap \CInt{r+r'}$$
\end{defn}

\begin{rmks}
$ $

\begin{enumerate}
  \item The definition of forcing orthogonal is dependent of the choice of
  the linear forcing structure, since it depends of the kind $\kappa$ and
  the choice of the predicate $\CInt{.}$. 
  \item Notice that the polarity of the predicate $\forcingOrtho{Z} : \kappa \rightarrow \sortpropNeg$
  does not depend of the polarity of the predicate $Z$: if $Z$ is a positive or negative
  predicate on $\kappa$, then $\forcingOrtho{Z}$ is a negative predicate on $\kappa$.
  This is a consequence of our choice of a negative encoding of the $\multimap$ connective. 
\end{enumerate}
\end{rmks}
We now define the forcing translation. We suppose that we associate to every 
\MAL variable $X$ a \HOMAL variable $X^{\kappa \rightarrow \sortpropPos}$
of kind $\kappa \rightarrow \sortpropPos$. 
If $A : \sortprop^{\circ}$ (with $\circ \in \{ +, -\}$) is a \MAL formula, we define
a \HOMAL constructor $A^{*} : \kappa \rightarrow \sortprop^{\circ}$  inductively as follows:
 \begin{eqnarray*}
 \preforc{X}   	 & \defequ & X^{\kappa\rightarrow \sortpropPos}\\
 \preforc{(X^{\bot})}  				 & \defequ & \forcingOrtho{X^{\kappa\rightarrow \sortpropPos}}\\
 \preforc{(A\otimes B)} & \defequ & \lambda r.\exists p_1.\exists p_2.\langle r = p_1 + p_2\rangle
 (\preforc{A}(p_1) \otimes \preforc{B}(p_2))\\
 \preforc{(A \parr B)} & \defequ & 
 			\forcingOrtho{\lambda r.\exists p_1.\exists p_2.\langle r = p_1 + p_2 \rangle (\preforc{(A^{\bot})}(p_1) \otimes \preforc{B^{\bot}}(p_2))}\\
\preforc{(\downarrow A)} & \defequ & 
		  \lambda r.\downarrow \preforc{A}(r) \\
 \preforc{(\uparrow A)} & \defequ & 
 			\forcingOrtho{\lambda r. \downarrow \preforc{(A^{\bot})}(r)}\\
\end{eqnarray*}
Finally, if $A$ is a \MAL formula and $p : \kappa$, we define $p\FVdash A$
as a \HOMAL constructor of kind $\sortpropNeg$ as follows:
\begin{eqnarray*}
 p \FVdash P & \defequ & \forcingOrtho{\forcingOrtho{P^{*}}}(p)\\
 p \FVdash N & \defequ & N^{*}(p)
\end{eqnarray*}

\begin{rmks}
$ $
\begin{enumerate}
\item
Informally, $A^{*}$ and $p\FVdash A$ have respectively the same role as the sets 
$\Inte{A}$ and $\interp{A}$ defined in Subsection \ref{subsec:interpcons}.
\item 
The predicate $A^{*}$ has the same polarity as $A$. However, 
the formula $p\FVdash A$ is always negative, even if $A$ is positive. 
\end{enumerate}
\end{rmks}

In the formalization of the forcing orthogonal, we use a negative encoding of the 
$\multimap$ connective: this is an adaptation of the negative forcing translation defined in
\cite{Miquel2011Forcing,krivine2010realizability}.  
 We could have defined a positive forcing translation, but what really matters is
 the polarity of $\CInt{.}$. Whereas in \cite{Miquel2011Forcing,krivine2010realizability}
 $\CInt{.}$ is always negative, we allow it to be positive. 

\begin{prop}\label{prop:posforcing}
$ $  
\begin{enumerate}
\item
For every negative formula $N$, we have $N^{*}(p) \congE \forcingOrtho{{N^{\bot}}^{*}}(p)$.
\item
For every positive formula $P$ and every $p: \kappa$,
 we have $$p\FVdash P \congE \forcingOrtho{\lambda r. (r\FVdash P^{\bot})}(p)$$
\end{enumerate}
\end{prop}

\subsection{The countdown machine}

We now describe a new abstract machine.
Because of the mechanism it implements, we call it 
the \textit{countdown machine}. Although this machine
is based on the term syntax of \LFOC, it has 
completely different reduction rules and hence
\emph{is not} just another extension of \LFOC.
We now suppose that the set $\Kpos$ contains 
new instructions constant $\overline{n}$ for each
$n\in \N$. Hence term syntax is augmented with 
\emph{primitive integers}. We denote by 
$\EncNat$ the set $\setp{\overline{n}}{n\in\N}$. 

\begin{defn}
To describe the evaluation in this machine,
we need to consider two new kinds of commands:

\begin{enumerate}
\item \textit{Negative commands} are of the form
$\com{\forcNeg{t}}{u^+}$ where $t$ is a positive or negative term whereas $u^+$ is a positive term.
The set of negative commands is denoted by $\forcNeg{\ComSet}$.
\item   \textit{Forcing commands}  are negative commands of the form
$\com{\forcNeg{t}}{(u,K)}$, where:
\begin{itemize}
  \item $K \in \Terms$ is a term (either positive or negative). 
  \item $t, u\in \Terms$ are terms of opposite polarities.
\end{itemize}
Such a forcing command will be sometimes noted  $\com{\forcNeg{t}}{u}\{K\}$.
\end{enumerate}
\end{defn}

If $K \in \Terms$,
$t^+ \in \PosTerms$ and $u^- \in \NegTerms$ are
respectively a positive and a negative term, we can build from the command $\com{t^+}{u^-}$
a forcing command noted $\FTrans{\com{t^+}{u^-}}\{K\} = \com{\forcNeg{u^-}}{( t^+,K)}$.

\begin{rmk}
The notation $\forcNeg{t}$ is just a marker on $t$ to indicate that 
it is now considered as negative. This is due to the formula translation.
Indeed, all the translated formulas are negative, so even if a term $t$
is positive, its \textit{image} in the machine is negative and 
can be executed in front of a positive term. 
\end{rmk}

The reduction relation
in this machine is denoted by $\redTrans$, and is defined between 
forcing commands by the following rules:

\begin{tabular}{rcl}
&  &\\
 $\com{\forcNeg{t^-}}{{(\mu \alpha.c)}}\{\overline{n}\}$ & $\redTrans$ & $\FTrans{(c[t^-/\alpha])}\{\overline{n}\}$\\
 $\com{\forcNeg{(\mu x.c)}}{V^+}\{\overline{n}\}$ & $\redTrans$ & $\FTrans{(c[V^+/x])}\{\overline{n}\}$\\
 $\com{\forcNeg{(\mu (\kappa,\kappa').c)}}{(V_1, V_2)}\{\overline{n+1}\}$ & $\redTrans$ & $\FTrans{(c[V_1/\kappa, V_2/\kappa'])}\{\overline{n}\}$\\
 $\com{\forcNeg{(\mu \{\kappa\}.c)}}{\{V\}}\{\overline{n+1}\}$ & $\redTrans$ & $\FTrans{(c[V/\kappa])}\{\overline{n}\}$\\ 
 &  & \\
 $\com{\forcNeg{(\mu (\kappa,\kappa').c)}}{(V_1, V_2)}\{\overline{0}\}$ &  $\FTrans{\Uparrow}$& \\
 $\com{\forcNeg{(\mu \{\kappa\}.c)}}{\{V\}}\{\overline{0}\}$ & $\FTrans{\Uparrow}$ & \\
  & &\\
\end{tabular}

\begin{rmk}
These rules indeed implement a kind of countdown: each step makes the counter decrease, and 
if the counter equals $0$ then any step makes the machine diverge. 
\end{rmk}

In the same spirit of the identification of $\com{t}{u}$ and $\com{u}{t}$, 
we quotient the set of forcing commands by the following $\alpha$-equivalence:
$$\com{\forcNeg{u}}{t}\{K\} \equiv \com{\forcNeg{t}}{u}\{K\}$$
It must be remarked that if $c$ is a command and $K \notin \EncNat$, then 
$\FTrans{c}\{K\}$ never reduces for $\redTrans$. \newline

\begin{rmk}
In contrast with \cite{Miquel2011Forcing}
we don't define any program transformation to justify the 
reduction rules of the machine. The justification of the introduction
of the machine will be given \textit{a posteriori}, by a specific
linear forcing structure. 
\end{rmk}

\subsection{A countdown machine-based realizability model} \label{subsec:count_real}

We now describe the realizability interpretation that is induced by 
the countdown machine. It is a 
simple realizability, in the sense that it relates a term $t$ and a
\HOMAL constructor $T$. It will be used
to state the \textbf{connection lemma}. This realizability
relation is \emph{parametrized} by a set
 $\BotTrans$ of
forcing negative commands closed under anti $\redTrans$-evaluation:
$$\forall c,c'\in \forcNeg{\ComSet},\text{ if }c \,\redTrans \,c'\text{ and }c' \in \BotTrans\text{ then }c\in\BotTrans$$

\begin{rmk}
Because of the $\alpha$-equivalence on forcing commands,
the following equivalence holds:
$\com{\forcNeg{t}}{(u,K)} \in \BotTrans \Leftrightarrow \com{\forcNeg{u}}{(t,K)}\in\BotTrans$.
\end{rmk}

We now suppose that such a set $\BotTrans$ is fixed.
Based on this set, we want to define a simple realizability
interpretation. 
Since $\BotTrans$ is not a set of commands, it is impossible to reuse immediately 
the definitions of Section \ref{sec:QuantKrivine}. But we can define a new
a set $\BotTrans_0$ of commands (in the usual sense) out of it: 
$$\BotTrans_0\,\defequ\, \setp{\com{t^+}{u^-}}{\com{\forcNeg{u^-}}{t^+}\in\BotTrans}$$

Now, suppose we have fixed a propositional domain 
$\FTrans{\Dpos}$ and a total valuation $\rhoTrans$. 
Although $\BotTrans_0$ is not closed under anti-evaluation for $\redh$, and hence
is not saturated, we can still consider the interpretation it induces. 
We then obtain
an interpretation of \HOMAL constructor. For each constructor $T$, this interpretation
is denoted $\InteTrans{T}{\rhoTrans}$, and the associated realizability relation is denoted
$\VdashTrans$.
In particular we have: 
$$t \VdashTrans N \Longleftrightarrow  \text{ for every } u\VdashTrans N^\bot \text{ we have } \com{\forcNeg{t}}{u}\in\BotTrans$$

\begin{rmks}
$ $
\begin{enumerate}
\item If the interpretation of a \HOMAL constructor can still be defined, neither
the adequacy theorem with respect to \HOMAL nor the properties of 
Subsection \ref{subsec:sat_prop} are valid. 
\item
The adequacy result,
as stated in Subsection \ref{subsec:nonquantreal}, does not hold. However,
we can and will prove a different adequacy result stated using forcing.
\item
Finally, it has to be noted that if $P$ is a positive formula and $t\in \Terms$
is a term, then $t\in \InteTrans{P}{\rhoTrans}$ implies that $t$ is positive. However, 
if $N$ is a negative formula, $t \in \InteTrans{N}{\rhoTrans}$ does not
imply that $t$ is negative, as we will see in Subsection \ref{subsec:connection}.
\end{enumerate}
\end{rmks} 

While the identification $\com{t}{u} = \com{u}{t}$ is reminiscent of 
the involutivity of the linear negation, the new $\alpha$-equivalence corresponds
to an identification between a term of the form $\forcingOrtho{X} : \kappa \rightarrow \sortpropNeg$
and its forcing biorthogonal $\forcingOrtho{\forcingOrtho{\forcingOrtho{X}}}$.  Indeed,
\begin{prop}\label{prop:forcing_alpha}
Suppose $T : \kappa\rightarrow \sortpropPos$, then if $p : \kappa$,
the following holds:
\begin{enumerate}
  \item $t \VdashTrans T(p) \text{ implies } t\VdashTrans \forcingOrtho{\forcingOrtho{T}}(p)$
  \item $t \VdashTrans \forcingOrtho{\forcingOrtho{\forcingOrtho{T}}}(p) \text{ implies } t\VdashTrans \forcingOrtho{T}(p)$
 \end{enumerate}
\end{prop}
\begin{proof}
$ $
\begin{enumerate}
  \item
Suppose $t\VdashTrans T(p) $. Then take $K \in \InteTrans{\CInt{p+r}}{}$ for some $r : \kappa$
and $u\in \InteTrans{\forcingOrtho{T}(r)}{\rhoTrans}$. Then because
$\forcingOrtho{T}(r) = \forall r'.T(r')^\bot\parr\FInt{r+r'}$,
we have $\com{u^\ominus}{(t,K)}\in\BotTrans$.
But, by $\alpha$-equivalence, we know
that $\com{t^\ominus}{(u,K)} \in \BotTrans$.
Hence, $t \VdashTrans \forcingOrtho{\forcingOrtho{T}}(p)$.
  \item If $t \VdashTrans \forcingOrtho{\forcingOrtho{\forcingOrtho{T}}}(p)$, $K \in 
  \InteTrans{\CInt{p+r}}{}$, 
  $u \in\InteTrans{T(r)}{\rhoTrans}$, then by the previous point, $u\in \InteTrans{\forcingOrtho{\forcingOrtho{T}}(r)}{\rho}$, 
  so $\com{t^\ominus}{(u,K)}\in\BotTrans$ which concludes. 
 \end{enumerate}
\end{proof}

Hence, as a corollary we immediately obtain that the two following rules are adequate
with respect to the $\InteTrans{.}{}$ interpretation:
$$
\begin{array}{cc}
\dfrac{ \vdash t : T(p) \delb \Gamma}{\vdash t : \forcingOrtho{\forcingOrtho{T}}(p)\delb\Gamma} & \dfrac{\vdash t : \forcingOrtho{\forcingOrtho{\forcingOrtho{T}}}(p) \delb \Gamma}{\vdash t : \forcingOrtho{T}(p)\delb\Gamma}
\end{array}
$$

\subsection{A quantitative linear forcing structure} \label{subsec:quant_forc_struct}

We now describe a particular linear forcing structure. The relation 
obtained by composition of the forcing translation induced by this structure
and the realizability based on the countdown machine of Subsection \ref{subsec:count_real}
will be shown in next subsection to coincide with a quantitative realizability relation.
The structure considered is $(\iota, \CInt{.}, +, \zero)$ where:
\begin{itemize}
  \item The kind of conditions is $\sortint$, the kind of natural numbers.
  \item $\CInt{.}$ is a new predicate of kind $\iota \rightarrow \sortpropPos$. 
  \item $+$ is the usual addition on natural numbers, defined using $rec_\iota$:
  $$.\sp +\sp .\quad =\quad \lambda p^{\iota}\lambda q^{\iota}.rec_\iota \sp p\sp \foSucc \sp q$$
  \item $\zero$  is the corresponding individual. 
\end{itemize}

\begin{prop}
$(\iota, \CInt{.}, +, \zero)$ is a linear forcing structure.
\end{prop}
\begin{proof}
  Associativity, commutativity and neutrality of $\foZero$ with respect to $+$
  are easily checked. As an example, we show the neutrality of $\foZero$. 
  We first notice that using the rules of Figure \ref{fig:HOMALcongruence},
  $ \foZero + q \congE rec_\iota \sp \foZero \sp \foSucc \sp  q$.
But we also have 
  $rec_\iota \sp \foZero \sp \foSucc \sp q \congE q$. 
Hence, by transitivity of $\congE$ we have $\foZero + q  \congE q$. 
\end{proof}

As $\CInt{.}$ is a new (positive) predicate of kind $\iota\rightarrow \sortpropPos$,
we need to say what its realizability interpretation is.  For each valuation $\rhoTrans$, 
we pose:
$$\InteTrans{\CInt{.}}{\rhoTrans} =  p \in \N \mapsto \setp{\overline{n}}{p\leq n\wedge n\in \N}$$

Since this function does not depend of the valuation $\rhoTrans$, we will not write the $\rhoTrans$
and note it $\InteTrans{\CInt{.}}{}$. 

\begin{rmks}
$ $

  \begin{enumerate}
  \item We will often switch between concrete elements of $\iota$ and elements of $\N$. 
  As already mentioned in Section 1, we will denote $\foN{n}$ the element of kind $\iota$
  corresponding to the integer $n\in\N$, which avoids confusion. 
  
  \item For the interpretation to make sense, we need $\FTrans{\Dpos}$ to contain
  all the sets $\InteTrans{\CInt{p}}{}$ with $p : \iota$. From now on, we will only consider such $\FTrans{\Dpos}$. 
    \item This linear forcing structure can be compared to the quantitative monoid of integers
    described in Example \ref{ex:int_monoid}. Indeed, we will se in Subsection \ref{subsec:forcing_cand}
    that they play the same role. 
    \item Similarly, many quantitative monoids can be turned into linear forcing structures.
    For example, suppose that we have extended the language of kinds with a product kind
    $\kappa \times\kappa'$ and added the pair $(T,U)$ and projections $\pi_i$  constructors     
     (it is not difficult to see how to extend the realizability model to such a framework). 
    Then, the soft monoid can be described as a linear forcing structure defined as follows:
      \begin{itemize} 
      	\item The kind is $\kappa = \iota \times (\iota \rightarrow \iota)$
		\item The maximum $max$ of two elements of $\iota$ is easily defined using $rec_\iota$,
		and the addition $+_s$ of the soft monoid can then be defined using $max$:
		$$+_s = \lambda x^{\kappa}.\lambda y^{\kappa}. (max(\pi_1 x^{\kappa},\pi_1 y^{\kappa}),
		\lambda z^{\iota}.max(\pi_2 x^{\kappa} (z), \pi_2 y^{\kappa} (z)))$$
		
		\item $\InteTrans{\CInt{p}}{} = \setp{\overline{n}}{\norm{p}\leq n\wedge n\in \N}$ where
		$$\norm{.} = \lambda x^{\kappa}.(\pi_2 x^{\kappa})(\pi_1 x^{\kappa})$$
      \end{itemize}\end{enumerate}
\end{rmks}

\subsection{A connection theorem}\label{subsec:connection}

In this subsection, the connection between quantitative realizability
and forcing is set out in the form of a \textbf{connection theorem}, which states
that the composition of the forcing relation induced by the quantitative linear
structure of Subsection
\ref{subsec:quant_forc_struct} and 
countdown machine based realizability of Subsection \ref{subsec:count_real} yields
a quantitative realizability model of \class{MAL}.
We then use this result together with Theorem \ref{thm:adequacy1} to obtain 
an adequacy result for linear forcing. \newline

We suppose having fixed a set $\BotTrans$ which is closed under anti $\redTrans$-evaluation
and the associated set of commands $\BotTrans_0$. We also suppose having a propositional
domain $\FTrans{\Dpos}$ and a total valuation $\rhoTrans$. 
We suppose having fixed a set $\BotTrans$ which is closed under anti $\redTrans$-evaluation
and the associated set of commands $\BotTrans_0$. We also suppose having a propositional
domain $\FTrans{\Dpos}$ and a total valuation $\rhoTrans$.

\begin{defn}
We define the following objects: 
\begin{itemize}
  \item $\BotQuant = \setp{(c,p)}{\forall K\in \InteTrans{\CInt{\foN{p}}}{}, \forcNeg{c}\{K\} \in \FTrans{\Bot}}$
  \item $\rhoQuant(X) =\setp{(t,p)}{t \in {\rhoTrans(X^{\iota\rightarrow \sortpropPos})(p)}}$ 
  \item $(\Dpos)^{\circ} = \setp{X}{\exists Y \in (\FTrans{\Dpos})^{\N}\text{ such that }(t,p)\in X\Leftrightarrow t \in Y(\foN{p})} $
\end{itemize}
Hence, $\rhoQuant$ is a valuation, which is not total but defined on all \MAL variables. Since we are only interested in \MAL variables, but need a total one to reuse the interpretation defined in Section \ref{sec:QuantKrivine}, we will in fact consider a valuation whose restriction on \MAL variables is $\rhoQuant$ and identify it with $\rhoQuant$.
\end{defn}

Because $\BotTrans$ is closed under anti-$\redTrans$ evaluation, $\BotQuant$
yields a saturated quantitative pole, as witnessed by the following property. 

\begin{prop}\label{prop:connec_sat}
$(\N,\BotQuant, 1)$ is a saturated quantitative pole.
\end{prop}
\begin{proof}
Suppose $(c',p)\in\BotQuant$ and $c\redh c'$. 
Then we want to prove that for any $K\in \InteTrans{\CInt{\foN{p}+\foN{1}}}{}$, $c\{K\} \in \BotTrans$.
But $K = \overline{n}$ and $p+1 \leq n$ so $n= n'+1$ and $K = \overline{n'+1}$.
So $\FTrans{c}\{\overline{n'+1}\} \redTrans \FTrans{c'}\{\overline{n'}\}$ with
$\overline{n'} \in \InteTrans{\CInt{\foN{p}}}{}$, so $\FTrans{c'}\{\overline{n'}\}\in \BotTrans$ and
by anti-reduction property, we obtain the conclusion.
\end{proof}

\begin{prop}
$(\Dpos)^{\circ}$ is a positive propositional domain.
\end{prop}
\begin{proof}
It is clear since $\FTrans{\Dpos}$ is itself a positive propositional domain. 
\end{proof}

Since $\BotQuant$ is a quantitative pole, $\Dpos^{\circ}$ is a propositional domain and $\rhoQuant$
is a total valuation, we obtain a quantitative realizability interpretation of \class{MAL}$\omega$, as defined in 
Section \ref{sec:QuantKrivine}. We denote this new interpretation $\InteQuant{.}{\rhoQuant}$, and the associated realizability interpretation $\VdashQuant$ . 
\begin{rmk}
We now have three different interpretations of \MAL formulas $A$:
\begin{itemize}
  \item The quantitative interpretation $\InteQuant{A}{\rhoQuant}$, 
  which is a set of bounded terms. 
  \item The non-quantitative interpretation $\InteTrans{A}{\rhoTrans}$, which
  is based on the countdown machine and is a set of terms (in fact, bounded terms where
  the bound is an element of the trivial monoid $\{0\}$). 
  \item The forcing interpretation $A^{*}(p)$, which is a \HOMAL formula. 
\end{itemize}
\end{rmk}
All these interpretations are related through the following \textbf{connection lemma}:

\begin{lem}\label{lm:connection}
For every \MAL formula $C$, every positive propositional domain $\FTrans{\Dpos}$ and every total valuation $\rhoTrans$, we have
$$t \in \InteTrans{C^{*}(\foN{p})}{\rhoTrans} \Longleftrightarrow (t,p) \in \InteQuant{C}{\rhoQuant}$$
\end{lem}
\begin{proof}
The proof is carried out by induction on the formula $C$. For each
case we prove directly the equivalence.
\begin{itemize}
  \item  If $C = X$, 
  \begin{eqnarray*}
   t \in \InteTrans{X^{*}(\foN{p})}{\rhoTrans}& \Longleftrightarrow & t \in \InteTrans{X^{\iota\rightarrow \sortpropPos}}{\rhoTrans}(p)\\
  	& \Longleftrightarrow & t \in \rhoTrans (X^{\iota\rightarrow\sortpropPos})(p)\\
  	& \Longleftrightarrow & (t,p) \in \rhoQuant (X)\\
  	& \Longleftrightarrow & (t,p)\in\InteQuant{X}{\rhoQuant}
  \end{eqnarray*}

  \item  If $C = X^\bot$,
  \begin{eqnarray*}
  	& & t \in \InteTrans{(X^{\bot})^{*}(\foN{p})}{\rhoTrans} \\
  	& \Longleftrightarrow & t \in \InteTrans{\forall r. X^{*}(r) \multimap \CInt{\foN{p}+r}^{\bot}}{\rhoTrans} \\
  	& \Longleftrightarrow & \forall r\in \N,\forall u \in \rhoTrans(X^{\iota\rightarrow\sortpropPos})(r), \forall K\in \InteTrans{\CInt{\foN{p}+r}}{}, \com{\forcNeg{t}}{(u,K)}\in\BotTrans\\
  	& \Longleftrightarrow & \forall r\in \N,\forall (u,r) \in \rhoQuant(X), \forall K\in \InteTrans{\CInt{\foN{p}+r}}{}, \com{\forcNeg{t}}{(u,K)}\in\BotTrans\\
	  	& \Longleftrightarrow & (t,p)\in \rhoQuant(X)^\bot\\
  	& \Longleftrightarrow & (t,p) \in \InteQuant{X^\bot}{\rhoQuant}
  \end{eqnarray*}
    
      \end{itemize} 
  
  In the remaining cases, we do not write the
valuations $\rhoQuant$ and $\rhoTrans$ in the interpretations, since they play no
role here. 
  
  \begin{itemize}
  \item If $C = \downarrow A$,
    
  \begin{tabular}{rcl}
   $t \in \InteTrans{\downarrow A^{*}({\foN{p}})}{\rhoTrans}$&  $\Longleftrightarrow$ & $t = \{ t' \}\text{ and } t' \in \InteTrans{A^{*}(\foN{p})}{\rhoTrans}$\\
  	& $\Longleftrightarrow$ & $t = \{ t' \}\text{ and } (t',p)\in \InteQuant{A}{\rhoQuant}$\\
  	& $\Longleftrightarrow$ & $(t,p) \in \InteQuant{\downarrow A}{\rhoQuant}$
  \end{tabular}
 
  \item  If $C = \uparrow A$,
  
  \begin{tabular}{rcl}
  	& & $t \in \InteTrans{(\uparrow A)^{*} (\foN{p})}{}$\\
	& $\Longleftrightarrow$ & $t \in \InteTrans{\forall x^{\iota}.(\downarrow A^{\bot})^{*}(x) \multimap
	\CInt{\foN{p}+x}^{\bot}}{}$\\
  	& $\Longleftrightarrow$ & $\forall q\in \N, t \in \InteTrans{\downarrow (A^{\bot})^{*}(\foN{q}) \otimes \CInt{\foN{p+q}}}{}^{\bot}$\\
  	& $\Longleftrightarrow$ & $\forall q\in \N, \forall K\in \InteTrans{\CInt{\foN{p+q}}}{},\, \forall u \in \InteTrans{(A^\bot)^{*}}{}(q), 
  	 \com{\forcNeg{t}}{(\{u\},K)} \in \BotTrans$\\
  	& $\Longleftrightarrow$ & $\forall (u,q) \in \InteQuant{A^\bot}{}, \forall K\in \InteTrans{\CInt{\foN{p+q}}}{},
  	 \com{\forcNeg{t}}{(\{u\},K)} \in \BotTrans$\\
  	& $\Longleftrightarrow$ & $\forall (u,q) \in \InteQuant{A^\bot}{}, (t,p)\bot^{\circ} (\{u\},q)$\\
  	& $\Longleftrightarrow$ & $(t,p) \in \InteQuant{\uparrow A}{} = \InteQuant{\downarrow A^{\bot}}{}^{\bot}$
  \end{tabular}
  	
	  \item If $C = A\otimes B$,
  
  \begin{tabular}{rcl}
  	& & $t \in \InteTrans{(A \otimes B)^{*}(\foN{p})}{}$\\
	& $\Leftrightarrow$ & $\exists p_1,p_2\in \N, p= p_1+p_2\wedge t \in \InteTrans{A^{*}(\foN{p_1})}{} \otimes \InteTrans{B^{*}(\foN{p_2})}{}$\\
  	&$\Leftrightarrow$ & $\exists p_1,p_2\in \N, p = (p_1,p_2)\wedge t = (t_1, t_2)\wedge t_1 \in \InteTrans{A^{*} (\foN{p_1})}{}
	\wedge t_2\in \InteTrans{B^{*}(\foN{p_2})}{}$\\
	& $\Leftrightarrow$ & $\exists p_1, p_2 \in \N, p = (p_1, p_2) \wedge t = (t_1,t_2) \wedge (t_1, p_1) \in \InteQuant{A}{}
	\wedge (t_2,p_2) \in \InteQuant{B}{}$\\
	& $\Leftrightarrow$ & $(t,p) \in \InteQuant{A\otimes B}{}$
  \end{tabular}
    \item  If $C = A\parr B$,
  \begin{eqnarray*}
  	& & t \in \InteTrans{(A \parr B)^{*}(\foN{p})}{}\\
  	&\Leftrightarrow &\forall u_1 \in \InteTrans{(A^\bot)^{*}(\foN{p_1})}{},
								\forall u_2\in \InteTrans{(B^\bot)^{*}(\foN{p_2})}{}, 
								\forall K\in \InteTrans{\CInt{\foN{p+p_1+p_2}}}{}, \\
  	& & \com{\forcNeg{t}}{((u_1,u_2),K)} \in \BotTrans\\
  	& \Leftrightarrow & \forall K\in \InteTrans{\CInt{\foN{p+p_1+p_2}}}{}, \forall (u_1, p_1) \in \InteQuant{A^\bot}{},\forall (u_2,p_2) \in \InteQuant{B^\bot}{},\\
  	& & \com{\forcNeg{t}}{((u_1,u_2),K)} \in \BotTrans\\
  	& \Leftrightarrow & \forall (u_1,p_1)\in \InteQuant{A^\bot},\forall (u_2,p_2)\in \InteQuant{B^\bot}{}, (t,p)\bot ((u_1,u_2),p_1+p_2)\\
  	& \Leftrightarrow & (t,p)\in \InteQuant{A\parr B}{}
  \end{eqnarray*}
  	        \end{itemize}
  	
\end{proof} 

As a corollary of this lemma and of Property \ref{prop:forcing_alpha}, we 
obtain the following connection theorem. 

\begin{thm}[Connection theorem]
For every \MAL formula $C$ and for every $t\in \Terms$ and $p\in \N$, we have
$$t \VdashTrans (\foN{p} \FVdash C) \Longleftrightarrow (t,p) \VdashQuant C$$
\end{thm}
\begin{proof}
We use the previous theorem. Let's distinguish two cases, depending of the polarity of $C$. 
\begin{enumerate}
\item Suppose $C = P$ is positive. Then:
\begin{eqnarray*}
t\VdashTrans (\foN{p}\FVdash P) &  \Leftrightarrow & t \in \InteTrans{\forcingOrtho{\forcingOrtho{P^*}}(\foN{p})}{}\\
& \Leftrightarrow & t\in \InteTrans{\forall r.(\forcingOrtho{P^{*}}(r) \multimap \CInt{\foN{p}+r}^{\bot})}{}\\
& \Leftrightarrow & \forall r \in \N. \forall u \in \InteTrans{\forcingOrtho{P^{*}}(\foN{r})}{},
\forall K \in \InteTrans{\CInt{\foN{p+r}}}{}, \com{\forcNeg{t}}{(u,K)}\in \BotTrans\\
& \Leftrightarrow & \forall r \in \N. \forall u \in \InteTrans{\forcingOrtho{P^{*}}(\foN{r})}{},
\forall K \in \InteTrans{\CInt{\foN{p+r}}}{}, \com{\forcNeg{t}}{(u,K)}\in \BotTrans
\end{eqnarray*}
 Since $P$ is positive, we have $(P^{\bot})^{*} = \forcingOrtho{P^{*}}$.
 Hence:
\begin{eqnarray*}
& &  \forall r \in \N. \forall u \in \InteTrans{\forcingOrtho{P^{*}}(\foN{r})}{},
\forall K \in \InteTrans{\CInt{\foN{p+r}}}{}, \com{\forcNeg{t}}{(u,K)}\in \BotTrans\\
& \Leftrightarrow &  \forall (u,r) \in \InteQuant{P^{\bot}}{},
\forall K \in \InteTrans{\CInt{\foN{p+r}}}{}, \com{\forcNeg{t}}{(u,K)}\in \BotTrans
\end{eqnarray*}
But if $\com{t^{\ominus}}{(u,K)}\in \BotTrans$ and $u$ is a negative term, then
it implies that $t$ is a positive term. Indeed $\BotTrans$ is a set of forcing negative commands. 
Hence,
\begin{eqnarray*}
&  &  \forall (u,r) \in \InteQuant{P^{\bot}}{},
\forall K \in \InteTrans{\CInt{\foN{p+r}}}{}, \com{\forcNeg{t}}{(u,K)}\in \BotTrans\\
& \Leftrightarrow &   (t,p)\in \InteQuant{P^{\bot}}{}^{\bot}\\
& \Leftrightarrow &   (t,p)\VdashQuant P\\
\end{eqnarray*}
\item Suppose $C = N$ is negative. 
Then $\forcingOrtho{N^{*}} = \forcingOrtho{(N^{\bot})^{*}}$.
We have
\begin{eqnarray*}
t\VdashTrans (\foN{p}\FVdash N) &  \Longleftrightarrow & t \in \InteTrans{\forcingOrtho{\forcingOrtho{N^*}}(\foN{p})}{}\\
 &  \Longleftrightarrow & t \in \InteTrans{\forcingOrtho{\forcingOrtho{\forcingOrtho{{(N^{\bot})^*}}}}(\foN{p})}{}
 \end{eqnarray*}
 By Property \ref{prop:posforcing}, we have
 $$t \in \InteTrans{\forcingOrtho{\forcingOrtho{\forcingOrtho{{(N^{\bot})^*}}}}(\foN{p})}{} \Longleftrightarrow
t \in \InteTrans{\forcingOrtho{{(N^{\bot})^*}}(\foN{p})}{}$$
 By Property \ref{prop:forcing_alpha}, we also have the following equivalence:
 $$ t\in \InteTrans{\forcingOrtho{(N^{\bot})^{*}}(\foN{p})}{} \Longleftrightarrow t \in \InteTrans{N^{*}(\foN{p})}{}$$
 Finally, by Lemma \ref{lm:connection}, we obtain 
$$  t \in \InteTrans{N^{*}(\foN{p})}{} \Longleftrightarrow (t,p) \in \InteQuant{N}{}$$
 Since, $N$ is negative, $\InteQuant{N}{} = \InteQuant{N}{}^{\bot\bot}$
 and hence we obtain
 $$t\VdashTrans (\foN{p}\FVdash N) \Longleftrightarrow (t,p)  \VdashQuant N$$
\end{enumerate}
\end{proof}

\begin{rmks}
$ $
\begin{enumerate}
\item This theorem shows that positive terms $t$ can realize (in the sense of
the $\InteTrans{.}{}$ interpretation) a negative formula. Indeed, if $P$ is positive
and $p\in \N$ then $p\FVdash P$ is always a negative formula, realized by a positive
term. 
\item This last theorem show the connection between quantitative linear forcing and 
quantitative realizability. It says considering this specific linear forcing structure
\emph{inside} the countdown realizability model (for any choice of a pole) is equivalent
to a certain quantitative realizability on \class{MAL}. 
\item If we have shown that composing the linear forcing and countdown realizability induces
a quantitative realizability relation, the converse is not true. Indeed, a quantitative realizability
relation is not a priori equivalent to the composition of a certain forcing and a countdown
realizability. 
\end{enumerate}
\end{rmks}

By Property \ref{prop:connec_sat}, we know that $\BotQuant$ is a saturated quantitative
pole. Hence, Theorem \ref{thm:adequacy1} is valid. 
Together with the connection theorem, it can be used to  obtain an adequacy theorem
for linear forcing with respect to \class{MAL}, \textit{inside} 
the realizability model. 

\begin{thm}\label{thm:forcing_adequacy}
Suppose $A, B_1, \dots, B_n$ are \MAL formulas.
Suppose $\pi$ is a proof
of $(\vdash t : A \delb \kappa_1 : B_1, \dots, \kappa_n : B_n)$.
Let $u_1,\dots, u_n \in \Terms$ and $q_1,\dots, q_n\in \N$
such that for any $i\in \Val{1,n}$ we have $u_i \VdashTrans (q_i \FVdash B_i)$.
Then if we pose $p = \Mp{\pi}$, we have 
$$t[u_1/\kappa_1, \dots, u_n/\kappa_n] \VdashTrans (\foN{p}+\foN{q_1}+\dots+\foN{q_n} \FVdash A)$$ 
\end{thm}

This result justifies \textit{a posteriori} the introduction of the countdown machine. \newline

\subsection{Forcing and reducibility candidates} \label{subsec:forcing_cand}

We have seen that it is possible to obtain certain instances of the quantitative
realizability by composing forcing and countdown machine-based realizability. We now show an example of such an instance: the quantitative reducibility
candidates of Subsection \ref{subsec:quant_cand} restricted to \class{MAL}, arise
from the composition of our quantitative linear forcing of Subsection \ref{subsec:quant_forc_struct}
and non quantitative reducibility candidates, adapted to the countdown machine. 

We choose a set $\BotTime$ of bounded commands, a positive propositional domain
$\Dpos$ and a total valuation $\rho$ such that:
\begin{eqnarray*}
\BotTime & =   &\setp{(c,n)}{c \text{ normalizes for }\redh^*\text{ using at most }n\,\,\beta\text{-steps}}\\
\rho(X^{\sortpropPos}) & = & \setp{(\daiPos, p)}{p\in \N}\\
\end{eqnarray*}
Those correspond to a particular case of quantitative reducibility candidates of Subsection \ref{subsec:quant_cand}. Indeed, in the case of the quantitative monoid
of integers, $\rho(X^{\sortpropPos})$ is a quantitative reducibility candidate.
We denote the corresponding realizability relation $\Vdash_{Time}$.

On another hand, we can define the (non-quantitative) reducibility candidates model
corresponding to the countdown machine. It is an instance of
the countdown machine based realizability. We choose the following $\BotTrans$ and 
the valuation $\rhoTrans$:
\begin{eqnarray*}
\BotTrans       & = & \setp{c\in {\ComSet^{\ominus}}}{c\text{ normalizes for } \redTrans}\\
\rhoTrans(X^{\iota\rightarrow \sortpropPos})(p) & = & \{\daiPos\}
\end{eqnarray*}

It is clear that these definitions correspond to those of Subsection \ref{subsec:quant_cand} where
the quantitative part has been erased. 
\begin{prop}\label{prop:timeforc}
$ $
\begin{enumerate}
  \item We have $\BotTime = \BotQuant$.
  \item For every \MAL atom $X^\sortpropPos$, we have
  $\rhoQuant(X^\sortpropPos) = \rho(X^{\sortpropPos})$.
\end{enumerate}
\end{prop}
\begin{proof}
$ $
\begin{enumerate}
  \item If $t_+$ and $u_-$ are two terms, then
  $\com{t_+}{u_-}$ normalizes in a number of $\redBeta$-steps at most $p$ if and only if
  for any $n\in\N$ greater than $p$, $\com{(u^-)^{\ominus}}{(\bar{n},t^+)}$ 
  normalizes for $\redTrans$. 
  \item This is by definition.
\end{enumerate}
\end{proof}

As an immediate corollary of Property \ref{prop:timeforc} and Lemma \ref{lm:connection},
we obtain the following \textbf{decomposition theorem}:
	
\begin{thm}\label{th:timeforcth} 
If $A$ is a \MAL formula, then 
$$(t,n)\Vdash_\Time A \Longleftrightarrow t\VdashTrans (\foN{n} \FVdash A)$$
\end{thm}

\subsection{Remarks} 

Let's finish by a few remarks about this forcing decomposition, and the choices we have made.\newline

\subtitle{Exponentials and quantifiers} Although we have only treated \MAL in this section, 
we could extend these results to a system with exponentials, like \HOSAL. 
To do so, it suffices to give an interpretation of $!$ and $?$. The definition would be:
\begin{eqnarray*}
(!A)^{*}  & = & \lambda x^{\kappa}.{\forall q. \langle x = !q \rangle (!(A^{*}(q)))}\\
(?A)^{*}  & = & \lambda x^{\kappa}.{\forcingOrtho{\lambda r.\forall q. \langle r = !q \rangle (!((A^\bot)^{*}(q)))}}(x)
\end{eqnarray*}
where $!$ is a term of kind $\kappa \rightarrow\kappa$. Of course, properties of $\CInt{.}$ with
respect to $!$ would be needed. 
Concerning the quantifiers, if one wants to interpret them, it suffices to follow 
the construction described in \cite{Miquel2011Forcing}. \newline

\subtitle{The choice of the machine} In the countdown machine, when the counter reaches $0$, any $\beta$-step
makes it diverge. This is a necessary choice we have made in order to obtain
Property \ref{prop:timeforc}. 
We could have chosen any other behavior without loosing Lemma \ref{lm:connection} and Theorem \ref{thm:forcing_adequacy}.
In particular, we could define a machine that executes programs for a certain number
of steps and then gives the hand back to a given program. \newline

\subtitle{Program transformation} In contrast with \cite{Miquel2011Forcing}, we don't justify our machine by 
exhibiting a program transformation. We could give such a transformation, as it amounts to reveal
the term behind the proof of the quantitative monoid part of Theorem \ref{thm:adequacy1}.

\section{Conclusion}

We have proposed an abstract quantitative framework, built upon 
Krivine's classical realizability for system $\class{L}_{foc}$ and the notion
of resource monoid developed in \cite{DalLago20112029}. As a
particular case of our construction, we have defined the
quantitative reducibility candidates, which allow us to
prove complexity properties of typable programs. 
Inside \class{MAL}$\omega$, we then have defined a linear forcing interpretation of \MAL
and an abstract machine that internalizes the computational behavior of the programs 
obtained through a particular instance of this forcing. We
finally have proved a decomposition theorem
which states that quantitative reducibility candidates for \MAL
can be obtained as the composition of ordinary reducibility candidates 
and linear forcing. \newline

We plan to explore several research directions.\newline

\subtitle{Order sensitive realizability} Both classical realizability and
	quantitative classical realizability are insensitive to order, in the logical
	sense. Indeed, unlike reducibility candidates, these techniques are 
	designed precisely to interpret second-order or higher-order logics. 
	In \cite{brunel23church}, a resource sensitive realizability is defined.
	One particularity of this realizability is that it cannot be used to
	interpret second-order quantifiers (in the paper, only a \textit{linear}
	second-order quantifier is interpreted), and thus allows an even finer
	grained study of the complexity properties of programs. This is achevied
	by using \textit{typed} abstract bounds. It would be interesting to see
	if this framework and ours can be both generalized into a new one: it could
	lead to a even more precise quantitative analysis. Indeed, if we are able
	to study the complexity due to the presence of different exponentials, what
	about the complexity due to quantifiers?\newline

\subtitle{Countdown and Implicit Complexity} Several attempts have been made to 
prove fundamental complexity results inside a purely logical framework. We
can for example mention the work of Terui and al. \cite{teruibassaurin} 
where a link between focalization and space compression theorem is
stated. 
In \cite{BaillotElementary}, the question is raised of whether it is possible
to prove in their setting a hierarchy theorem like $\class{P} \subsetneq \class{EXP}$. It is striking that such an elementary complexity result
cannot be easily proved in a proof theoretic setting. 
We conjecture that one of the main reasons for this apparent difficulty 
is the lack of expressivity of the logics at stake.
For example, all known proofs of $\class{P}\subsetneq \class{EXP}$
crucially rely on defining a Turing machine which executes another Turing
 machine on a entry for a certain number of steps. 
 This kind of feature is not available in a purely functional $\lambda$-calculus:
 it is not possible to internalize such a $\lambda$-evaluator inside
 the typed $\lambda$-calculus. Moreover, adding this feature would rise the problem of how to type 
 programs using it.  
 This is exactly the functionality implemented by the countdown machine. Using
 a variant of it, we would be able to execute programs for a certain 
 number of steps. Typing those terms can be achieved using forcing. 
 It seems to us we could use these facts to prove hierarchy theorems in
 a purely proof theoretic setting.
 As a first test for this idea, we plan to see whether is possible
 to prove $\class{P}\subsetneq \class{EXP}$ in a
 a forcing extension 
 of the system described in \cite{BaillotElementary}.  \newline
 
\subtitle{Krivine Linear Algebras}  We plan to reformulate this framework in a more abstract
 setting, in the spirit of Krivine's realizability
 algebras \cite{krivine2010realizability}. In such a framework, we could express both 
 quantitative realizability and linear forcing. We could hopefully prove a general iteration
 theorem of which our connection lemma would be a particular case.  \newline
 
	\subtitle{A Logic of forcing}  	As already noticed, while the identification 
  	$\com{t}{u} \equiv \com{u}{t}$ made in $\class{L}_{foc}$ materializes
  	the involutivity of negation,   	
  	the identification \linebreak {$\com{\forcNeg{u}}{t}\{K\} \equiv \com{\forcNeg{t}}{u}\{K\}$}
  	accounts for the properties of forcing orthogonality. 
  	This remark suggests a new logic where forcing orthogonality
  	would be \textit{primitive}, just like negation in linear logic. 
  	In such a setting, forcing would be easily recovered 
  	and dealt directly with.   
  We are currently investigating a \textit{logic of forcing}, which could be used 
  as a type system for a calculus with effects.  \newline
  
  \subtitle{Differential privacy and function sensitivity} \textit{Differential privacy}
  \cite{dwork2006differential} is a quantitative property of randomized functions (typically functions
  giving an answer to user queries on a database) that prevents 
  malicious users to gain confidential knowledge from repeated queries. 
  In \cite{reed2010distance}, a linear type system that ensures differential privacy is proposed. It is based
  on \textit{function sensitivity}, a measure of how the \textit{distance} between outputs
  of a function is related to the distance between the respective inputs (this property
  is similar to \textit{Lipschitz continuity}). We plan to see if the logical relations used in \cite{reed2010distance}
  to prove the soundness of their type system can be reformulated in terms
  of quantitative realizability.
  
%
%
  	

\bibliographystyle{plain}
\bibliography{biblio}

\begin{thebibliography}{10}

\bibitem{BaillotElementary}
P.~Baillot.
\newblock {Elementary linear logic revisited for polynomial time and an
  exponential time hierarchy (extended version)}.
\newblock {\em to appear in the Proceedings of Asian Symposium on Programming
  Languages and Systems (APLAS 2011)}, 2011.

\bibitem{baillot2004soft}
P.~Baillot and V.~Mogbil.
\newblock Soft lambda-calculus: a language for polynomial time computation.
\newblock In {\em Foundations of software science and computation structures},
  pages 27--41. Springer, 2004.

\bibitem{brunel23church}
A.~Brunel and K.~Terui.
\newblock {Church $\Rightarrow$ Scott = Ptime: an application of resource
  sensitive realizability}.
\newblock {\em Electronic Proceedings in Theoretical Computer Science}, 23,
  2010.

\bibitem{cohen1963independence}
P.J. Cohen.
\newblock {The independence of the continuum hypothesis}.
\newblock {\em Proceedings of the National Academy of Sciences of the United
  States of America}, 50(6):1143, 1963.

\bibitem{curien2000duality}
P-L. Curien and H.~Herbelin.
\newblock The duality of computation.
\newblock In {\em ACM sigplan notices}, volume~35, pages 233--243. ACM, 2000.

\bibitem{dwork2006differential}
C.~Dwork.
\newblock Differential privacy.
\newblock {\em Automata, languages and programming}, pages 1--12, 2006.

\bibitem{girard1987linear}
J-Y. Girard.
\newblock {Linear logic}.
\newblock {\em Theoretical computer science}, 50(1):1--101, 1987.

\bibitem{CambridgeJournals:4439772}
J-Y. Girard.
\newblock A new constructive logic: classic logic.
\newblock {\em Mathematical Structures in Computer Science}, 1(03):255--296,
  1991.

\bibitem{DBLP:journals/apal/Hofmann00}
M.~Hofmann.
\newblock Safe recursion with higher types and bck-algebra.
\newblock {\em Ann. Pure Appl. Logic}, 104(1-3):113--166, 2000.

\bibitem{Hofmann2004121}
M.~Hofmann and P.J. Scott.
\newblock Realizability models for bll-like languages.
\newblock {\em Theoretical Computer Science}, 318(1-2):121 -- 137, 2004.
\newblock Implicit Computational Complexity.

\bibitem{kleene1969formalized}
S-C. Kleene.
\newblock {\em Formalized recursive functionals and formalized realizability},
  volume~89.
\newblock Amer Mathematical Society, 1969.

\bibitem{krivine-realizability}
J-L. Krivine.
\newblock {Realizability in classical logic. Course notes of a series of
  lectures given in the University of Marseille, may 2004 (last revision: july
  2005)}.
\newblock {\em Panoramas et syntheses, Soci{\'e}t{\'e} Math{\'e}matique de
  France}, 2005.

\bibitem{krivine2007call}
J-L. Krivine.
\newblock A call-by-name lambda-calculus machine.
\newblock {\em Higher-Order and Symbolic Computation}, 20(3):199--207, 2007.

\bibitem{krivine2010realizability}
J-L. Krivine.
\newblock {Realizability algebras: a program to well order R}.
\newblock {\em manuscript}, 2010.

\bibitem{krivine2010realizability2}
J-L. Krivine.
\newblock {Realizability algebras II: new models of ZF+ DC}.
\newblock {\em Arxiv preprint arXiv:1007.0825}, 2010.

\bibitem{dalagobll}
U.~Dal Lago and M.~Hofmann.
\newblock Bounded linear logic, revisited.
\newblock In Pierre-Louis Curien, editor, {\em Typed Lambda Calculi and
  Applications}, volume 5608 of {\em Lecture Notes in Computer Science}, pages
  80--94. 2009.

\bibitem{dal2010semantic}
U.~Dal Lago and M.~Hofmann.
\newblock A semantic proof of polytime soundness of light affine logic.
\newblock {\em Theory of Computing Systems}, 46:673--689, 2010.

\bibitem{DalLago20112029}
U.~Dal Lago and M.~Hofmann.
\newblock Realizability models and implicit complexity.
\newblock {\em Theoretical Computer Science}, 412(20):2029 -- 2047, 2011.
\newblock Girard's Festschrift.

\bibitem{LengrandMiquelClassical2008}
S.~Lengrand and A.~Miquel.
\newblock {Classical F [omega], orthogonality and symmetric candidates}.
\newblock {\em Annals of Pure and Applied Logic}, 153(1-3):3--20, 2008.

\bibitem{teruibassaurin}
A.~Saurin M.~Basaldella and K.~Terui.
\newblock On the meaning of focalization.
\newblock 6505:78--87, 2011.

\bibitem{miquel2007classical}
A.~Miquel.
\newblock Classical program extraction in the calculus of constructions.
\newblock In {\em Computer Science Logic}, pages 313--327. Springer, 2007.

\bibitem{Miquel2011Forcing}
A.~Miquel.
\newblock Forcing as a program transformation.
\newblock In {\em Logic in Computer Science (LICS), 2011 26th Annual IEEE
  Symposium on}, pages 197--206. IEEE, 2011.

\bibitem{MunchReal}
G.~Munch-Maccagnoni.
\newblock Focalisation and classical realisability.
\newblock In {\em Computer Science Logic}, pages 409--423. Springer, 2009.

\bibitem{okada1999phase}
M.~Okada.
\newblock {Phase semantic cut-elimination and normalization proofs of first-and
  higher-order linear logic}.
\newblock {\em Theoretical Computer Science}, 227(1-2):333--396, 1999.

\bibitem{pagani2010strong}
M.~Pagani and L.~Tortora de~Falco.
\newblock {Strong normalization property for second order linear logic}.
\newblock {\em Theoretical Computer Science}, 411(2):410--444, 2010.

\bibitem{reed2010distance}
J.~Reed and B.C. Pierce.
\newblock Distance makes the types grow stronger: A calculus for differential
  privacy.
\newblock In {\em ACM SIGPLAN Notices}, volume~45, pages 157--168. ACM, 2010.

\end{thebibliography}

\end{document}